\documentclass[a4paper, UKenglish, cleveref, autoref, thm-restate,numberwithinsect]{lipics-v2021}

\usepackage{mathtools}
\usepackage{colortbl}
\usepackage{algorithm}
\usepackage{bbm}

\hideLIPIcs

\newcommand{\reffig}[1]{Figure~\ref{fig:#1}}

\newcommand{\reflem}[1]{Lemma~\ref{lem:#1}}

\newcommand{\refalg}[1]{Algorithm~\ref{algo:#1}}
\DeclareMathOperator{\poly}{poly}

\newcommand{\restricteddiagonals}{diagonal-restricted}

\newcommand{\enquote}[1]{``#1''}

\newcommand{\abs}[1]{\left\vert#1\right\vert}
\newcommand{\N}{\mathbb{N}}
\newcommand{\I}{\mathscr{I}}
\newcommand{\J}{\mathscr{J}}

\newcommand{\eps}{\varepsilon}

\newcommand{\metricspace}{M}
\newcommand{\distance}{d}
\newcommand{\traversal}{\mathcal{C}}
\newcommand{\df}[2]{D_{\mathcal{F}}\left(#1,#2\right)}
\newcommand{\hd}[2]{D_{\mathcal{H}}\left(#1,#2\right)}

\newcommand{\tloc}{t}
\newcommand{\tlocal}{t\text{-local}}

\newcommand{\fsm}{M}
\newcommand{\firstcurve}{P}
\newcommand{\secondcurve}{Q}

\newcommand{\vertfirstcurve}{p}
\newcommand{\vertsecondcurve}{q}
\newcommand{\numvertfirst}{n}
\newcommand{\numvertsecond}{m}

\newcommand{\optcoupling}{\mathcal{C}}

\newcommand{\n}{\numvertfirst}

\newcommand{\succprob}{\frac{4}{5}}

\definecolor{ashgrey}{rgb}{0.85, 0.85, 0.85}

\newcolumntype{q}{>{\columncolor{pink}}c}
\newcolumntype{g}{>{\columncolor{ashgrey}}c}

\title{Property Testing of Curve Similarity}

\titlerunning{Property Testing of Curve Similarity}
\author{Peyman Afshani}{Department of Computer Science, Aarhus University, Denmark}{peyman@cs.au.dk}{https://orcid.org/0000-0001-6102-0759}{}
\author{Maike Buchin}{Faculty of Computer Science, Ruhr University Bochum, Germany}{maike.buchin@rub.de}{https://orcid.org/0000-0002-3446-4343}{}
\author{Anne Driemel}{Institute for Computer Science, University of Bonn, Germany}{driemel@cs.uni-bonn.de}{https://orcid.org/0000-0002-1943-2589}{Affiliated with Lamarr Institute for Machine Learning and Artificial Intelligence.}
\author{Marena Richter}{Institute for Computer Science, University of Bonn, Germany}{marenarichter@uni-bonn.de}{https://orcid.org/0009-0007-8250-266X}{}
\author{Sampson Wong}{Department of Computer Science, University of Copenhagen, Denmark}{sampson.wong123@gmail.com}{https://orcid.org/0000-0003-3803-3804}{}

\authorrunning{P. Afshani, M. Buchin, A. Driemel, M. Richter, S. Wong}

\Copyright{Peyman Afshani, Maike Buchin, Anne Driemel, Marena Richter, Sampson Wong}



\category{} 

\relatedversion{} 

\acknowledgements{Anne Driemel would like to thank Ronitt Rubinfeld and Ilan Newman for helpful discussions in the early stages of this research. The authors thank Sariel Har-Peled for feedback on an earlier version of this manuscript.}

\nolinenumbers

\EventEditors{John Q. Open and Joan R. Access}
\EventNoEds{2}
\EventLongTitle{42nd Conference on Very Important Topics (CVIT 2016)}
\EventShortTitle{CVIT 2016}
\EventAcronym{CVIT}
\EventYear{2016}
\EventDate{December 24--27, 2016}
\EventLocation{Little Whinging, United Kingdom}
\EventLogo{}
\SeriesVolume{42}
\ArticleNo{23}

\renewcommand{\kappa}{t}
\begin{document}

\maketitle

\begin{abstract}
We propose sublinear algorithms for probabilistic testing of the discrete and continuous Fréchet distance---a standard similarity measure for curves. We assume the algorithm is given access to the input curves via a query oracle: a query returns the set of vertices of the curve that lie within a radius $\delta$ of a specified vertex of the other curve.
The goal is to use a small number of queries to determine with constant probability whether the two curves are similar (i.e., their discrete Fréchet distance is at most $\delta$)
or they are \enquote{$\eps$-far} (for $0 < \eps < 2$) from being similar, i.e., more than an $\eps$-fraction of the two curves must be ignored for them to become similar. 
We present two algorithms which are sublinear assuming that the curves are $\kappa$-approximate shortest paths in the ambient metric space, for some $\kappa \ll n$. The first algorithm uses $O(\frac{\kappa}{\eps} \log\frac{\kappa}{\eps})$  queries and is given the value of $\kappa$ in advance. The second algorithm does not have explicit knowledge of the value of $\kappa$ and therefore needs to gain implicit knowledge of the straightness of the input curves through its queries. We show that the discrete Fréchet distance can still be tested using roughly $O(\frac{\kappa^3+\kappa^2 \log n}{\eps})$ queries ignoring logarithmic factors in $\kappa$. 
Our algorithms work in a matrix representation of the input and may be of independent interest to matrix testing. 
Our algorithms use a mild uniform sampling condition that constrains the edge lengths of the curves, similar to a polynomially bounded aspect ratio. Applied to testing the continuous Fréchet distance of $\kappa$-straight curves, our algorithms can be used for $(1+\eps')$-approximate testing using essentially the same bounds as stated above with an additional factor of $\poly(\frac{1}{\eps'})$. 
\end{abstract}

\newpage
\tableofcontents
\newpage

\pagenumbering{arabic} 

\section{Introduction} 

We initiate the study of property testing for measures of curve similarity,
motivated by the need for fast solutions for curve classification and clustering. Thus,
our research lies at the intersection of property testing and computational geometry.  
While property testing is a very broad area, we believe this intersection has received far less attention than
it deserves. We also believe that property testing for well-studied measures such as the Fr\' echet 
distance is especially well-motivated due to its connections to other problems studied on the curves, such as clustering~\cite{DBLP:conf/gis/BrankovicBKNPW20,DBLP:conf/gis/BuchinDLN19}, similarity search~\cite{DBLP:conf/gis/BaldusB17,DBLP:conf/gis/BuchinDDM17,DBLP:conf/gis/DutschV17} and map reconstruction~\cite{DBLP:conf/gis/BuchinBDFJSSSW17,DBLP:conf/gis/BuchinBGHSSSSSW20}.

Typically in property testing, we are given access to a (large) data set and the goal is to very quickly assess whether
the data has a certain property.
The classical notation of correctness is typically not suitable for property testing since
these algorithms
are intended to be randomized algorithms that access
(very small) parts of data via queries.
Instead, a property testing algorithm is considered
correct if it can satisfy the following two conditions, with a probability close
to 1 (e.g., at least $1-\gamma$ probability for some small value of $\gamma$): 
first, if the input has the desired property, the algorithm must return \textit{accept}  
and second, if the input is \enquote{far} from having the property (under some suitable definition of \enquote{far}), the 
algorithm should \textit{reject} the input. 
Our algorithms have \textit{one-sided error}, meaning, they will not produce \textit{false negatives}
and thus rejection means that the input does not have the desired property.

For further reading on property testing, see~\cite{dana.survey, bhattacharyya2022property}. 
Here, we quickly review the main motivation for property testing in general and then we focus on 
computational geometry. 

Property testing algorithms can be useful in many different scenarios.
In particular, if the input is extremely large, it makes sense to obtain a quick
approximate answer before deciding to run more expensive algorithms. 
This is especially useful if the testing algorithm has one-sided error; in that case, if the test returns
\textit{reject}, there is no need to do any extra work. For the Fr\'echet distance, there are applications~\cite{DBLP:conf/gis/BaldusB17,DBLP:conf/gis/BuchinDDM17,DBLP:conf/gis/DutschV17} where negative filters are used to minimize expensive Fr\'echet distance computations; testing may be useful in such cases.
Property testing is also useful if most inputs are expected to not have the desired property, or 
the input distribution is such that each input either has the property or it is \enquote{far} from having that
property.
In addition, the definition of property testing has a strong connection to topics in learning theory such as
\textit{probably approximately correct (PAC)} learners  and this was in fact one of the important motivations
behind its conception~\cite{ron2001property}. 

Another motivation for property testing is when small errors can be tolerated or objects that are close to 
having the desired property are acceptable outcomes; in this view, property
testing is in spirit similar to approximation algorithms where an approximate answer is considered to be an acceptable output. In the context of similarity testing we can also compare it to a partial similarity measure. In fact, our chosen error model will be very close to a partial Fréchet distance~\cite{buchin2009exact}. For more details on motivations to study property testing see~\cite{goldreich.book,dana.survey}.
See also the aforementioned references for more detailed results on property testing as 
surveying the known results on property testing is beyond the scope of this paper.

Computational geometry has a long tradition of using randomization and sampling to speed up algorithmic approaches~\cite{Mulmuley, sariel-book,matousek-book, handbook}. Property testing has received some attention within computational geometry, but is much less explored compared to other areas.
There are fast and efficient testers for many basic geometric properties, such as convexity of a point set~\cite{czumaj2000property},
disjointness of objects~\cite{czumaj2000property}, the weight of Euclidean minimum spanning tree~\cite{ chazelle05,CzumajEFMNRS05,czumaj09}, clustering~\cite{monemizadeh2023}, the Voronoi property of a triangulation~\cite{czumaj2000property}, ray shooting~\cite{chazelle.sublinear,czumaj2000property}, as well as LP-type problems~\cite{DBLP:conf/icalp/EpsteinS20}. 

Unlike in graph property testing, where it is assumed that we can sample a vertex uniformly at random, and query for its neighbors, it seems that, for geometric data, there is no commonly agreed upon query model. Different types of queries are used to access the data in the above cited works, such as range queries, cone nearest neighbor queries,  queries for the bounding box of the data inside a query hyper-rectangle. 
In our paper, we use queries to the free space matrix of the two input curves, see Section~\ref{sec:model} for the precise definition. 

We mention that there is also work on matrix testing~\cite{fischer2001testing}. Compared to our setting, the input size is defined as $n^2$---the number of entries needed to specify the matrix. In our case, we assume the input size to be $n$, the maximum of column and row dimension of the matrix.

\section{Preliminaries}\label{sec:prelims}
Let $(\metricspace,\distance)$ be a metric space. We say a \emph{(discrete) curve} $\firstcurve$ in $(\metricspace,\distance)$ is an ordered point sequence $\langle\vertfirstcurve_1,\ldots,\vertfirstcurve_{\numvertfirst}\rangle$ with $\vertfirstcurve_i\in\metricspace$ for all $i=1,\ldots,\numvertfirst$. 
We call the points of the curve \emph{vertices}.
We denote by $\abs{\firstcurve}$ the number of vertices in $\firstcurve$ and by $\ell(P)$ its \emph{length}, which is defined as $\ell(P)=\sum_{i=1}^{\numvertfirst-1}\distance(\vertfirstcurve_i,\vertfirstcurve_{i+1})$.
The \emph{subcurve} of $\firstcurve$ between $\vertfirstcurve_i$ and $\vertfirstcurve_j$ is denoted by $\firstcurve[i,j]$.
A curve $\firstcurve$ is called \emph{$\kappa$-straight} if for any two vertices $\vertfirstcurve_i$ and $\vertfirstcurve_j$ in $\firstcurve$, we have $\ell(P[i,j])\leq\kappa\cdot\distance(\vertfirstcurve_i,\vertfirstcurve_j)$.
Denote by $[\n]$ the set of integers from $1$ to $\n$ and by $[\n]\times[\n]\subset\N\times\N$ the integer lattice of $\n$ times $\n$ integers.
Given two curves $\firstcurve=\langle\vertfirstcurve_1,\ldots,\vertfirstcurve_{\numvertfirst}\rangle$ and $\secondcurve=\langle\vertsecondcurve_1,\ldots,\vertsecondcurve_{\n}\rangle$, we say that an ordered sequence $\traversal$ of elements in $[\n]\times[\n]$ is a \emph{coupling} of $\firstcurve$ and $\secondcurve$, if it starts in $(1,1)$, ends in $(\n,\n)$ and for any consecutive tuples $(i,j),(i',j')$ in $\traversal$ it holds that {$(i',j')\in\{(i+1,j),(i,j+1),(i+1,j+1)\}$}.
We define the \emph{discrete Fréchet distance}
between $\firstcurve$ and $\secondcurve$ as follows 
\[\df{\firstcurve}{\secondcurve}\coloneqq\min\limits_{\text{coupling }\traversal}\max_{(i,j)\in\traversal}\distance(\vertfirstcurve_i,\vertsecondcurve_j).\]
For brevity, we simply call this the Fréchet distance between $\firstcurve$ and $\secondcurve$. 
One can verify that the Fréchet distance satisfies the triangle inequality.
The \emph{free space matrix} of $\firstcurve$ and $\secondcurve$ with distance value $\delta$ is an $\n\times\n$ matrix $\fsm_{\delta}$, where the $i$-th column corresponds to the vertex $\vertfirstcurve_i$ of $\firstcurve$ and the $j$-th row corresponds to the vertex $\vertsecondcurve_j$ of $\secondcurve$. 

The entry $\fsm_{\delta}[i,j]$ has the value $0$ if $\distance(\vertfirstcurve_i,\vertsecondcurve_j)\leq\delta$ and $1$ otherwise.\footnote{Note we use $0$ and $1$ in switched roles compared to the conventions in the literature on the free space matrix. Our notation makes the cost-function of the path more intuitive.}
A coupling path $\mathcal{C}$ is the path through the $[\n]\times[\n]$ integer lattice defined by the tuples of the coupling $\mathcal{C}$.
We define the \emph{cost} of such a path as $c(\mathcal{C})=\sum_{(i,j)\in\mathcal{C}}\fsm_{\delta}[i,j]$.  
Note that the Fréchet distance between $\firstcurve$ and $\secondcurve$ is at most $\delta$ if and only if there exists a coupling path with cost $0$ from $(1,1)$ to $(\n,\n)$. 

Our analysis is based on a property of the free space matrix. We first define this property and then link the property to a certain class of well-behaved input curves.

\begin{definition}[$\tlocal$]
    \label{definition:t-local}
    Let $\fsm$ be a free space matrix of curves $\firstcurve$ and $\secondcurve$. We say that $\fsm$ is \emph{$\tlocal$} if, for any tuples $(i_1,j_1)$ and $(i_2,j_2)$ with $\fsm[i_1,j_1]=0=\fsm[i_2,j_2]$, it holds that $\abs{i_1-i_2}\leq\tloc\cdot(2+\abs{j_1-j_2})$ and $\abs{j_1-j_2}\leq\tloc\cdot(2+\abs{i_1-i_2})$. 
\end{definition}

See Figure~\ref{fig:t-locality} for a visualization.
A consequence of this property is that zero-entries within one row or within one column cannot be too far away from each other. 
\begin{observation}\label{obs:ones-in-same-row-close}
    Suppose $M$ is $\tlocal$. If we have $\fsm[i,j]=0=\fsm[i,j']$, then we have $\abs{j-j'}\leq2\tloc$. If we have $\fsm[i,j]=0=\fsm[i',j]$, we have $\abs{i-i'}\leq2\tloc$.
\end{observation}

\begin{figure}
    \centering
    \includegraphics[width=0.55\linewidth]{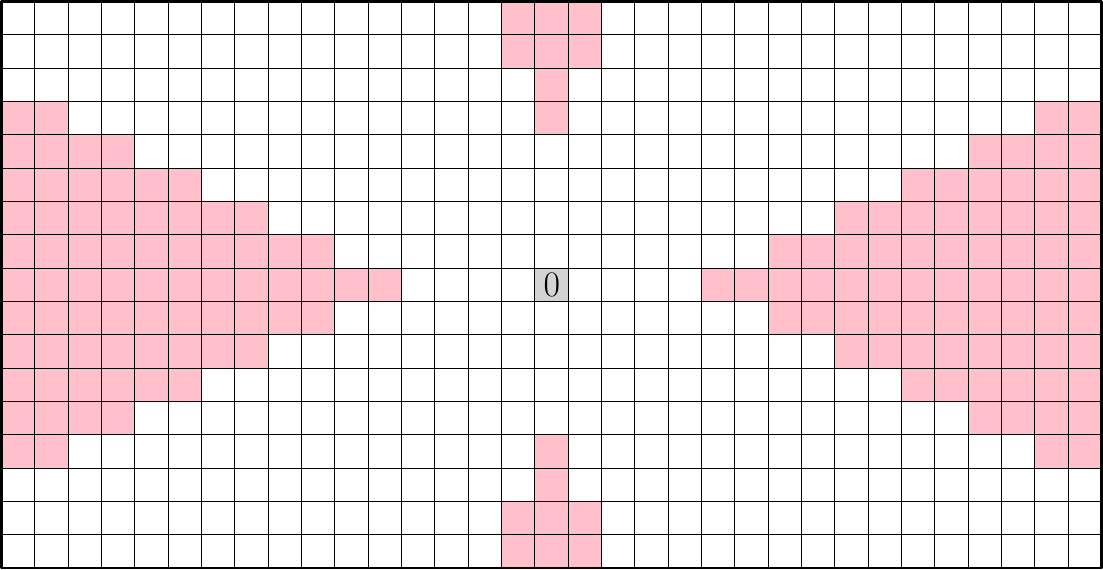}
    \caption{The red cells cannot be zero-entries if the gray cell is a zero-entry and the matrix is $2$-local.}
    \label{fig:t-locality}
\end{figure}

In the following, we argue why we think the $\tloc$-locality assumption is reasonable. Note that $\tloc$ will not exceed $n$ on any input. Moreover, we argue that there are classes of well-behaved curves that satisfy the locality assumption for $t \ll n$. Lemma~\ref{lem:straight} below shows that the free space matrix of approximate shortest paths is $t$-local.
For simplicity, we impose the assumption that the lengths of the edges are bounded by a fixed multiple of $\delta$, the query radius of the Fréchet distance. 
In Section~\ref{sec:largedelta}, we show that our approach can be adapted to work as well if the lengths of the edges are bounded by a constant multiple of any fixed value.

\begin{lemma}\label{lem:straight}
    Let $\firstcurve$ and $\secondcurve$ be $\kappa$-straight curves with edge lengths in $[\delta/\alpha,\alpha\delta]$ for some constant $\alpha\geq1$. Let $\fsm$ be their free space matrix with distance value $\delta$. 
    Then, $\fsm$ is $\mathcal{O}(\kappa)$-local.
\end{lemma}
\begin{proof}
    Let $(i_1,j_1),(i_2,j_2)$ be such that $\fsm[i_1,j_1]=0=\fsm[i_2,j_2]$ and $i_1\leq i_2$.
    So we have $d(p_{i_1},q_{j_1})\leq\delta$ and $d(p_{i_2},q_{j_2})\leq\delta$. We also have $d(q_{j_1},q_{j_2})\leq\ell(Q[j_1,j_2])\leq\alpha\delta\abs{j_1-j_2}$ because of the edge length constraint, where $Q[j_1,j_2]$ is replaced by $Q[j_2,j_1]$ if $j_2<j_1$. Using the edge length constraint and $\kappa$-straightness we derive $\frac{\delta}{\alpha}\abs{i_2-i_1}\leq\ell(P[i_1,i_2])\leq\kappa d(p_{i_1},p_{i_2})$.
    These observations together with the triangle inequality yield
    \[\abs{i_2-i_1}\leq\frac{\alpha\kappa}{\delta}(d(p_{i_1},q_{j_1})+d(q_{j_1},q_{j_2})+d(q_{j_2},p_{i_2}))\leq\alpha^2\kappa(2+\abs{j_2-j_1}).\]
    The other inequality is shown by reversing the roles of $P$ and $Q$.
\end{proof}

\section{Problem definition and results}\label{sec:model}

\subsection{The basic problem}
The problem we study in this paper is the following. Assume we want to determine for two curves $\firstcurve$ and $\secondcurve$, each consisting of $n$ vertices,\footnote{For ease of notation, our analysis assumes the input curves have the same number of vertices. This assumption can be removed.} if their Fréchet distance is at most a given value of $\delta$. We do not have direct access to the input curves. Instead, we have access to an oracle that returns the information in a given row or column  of the $\delta$-free space matrix in the form of a sorted list of indices of zero-entries. We call this a \emph{query} and we want to determine $\df{\firstcurve}{\secondcurve}\leq\delta$ using as few (sublinear in $n$) queries as possible. Note that from the point of view of a data structure setting, our query corresponds to a classical ball range query with a vertex $p$ of one curve and returns the list of vertex indices of the other curve that are contained in the ball of radius $\delta$ centered at $p$. 
 
Our bounds on the number of queries will be probabilistic and will hold under a certain error model. The error model allows for the coupling path to pass through a bounded number of one-entries of the free space matrix. These are entries where the corresponding points are far from each other. In Section~\ref{sec:errormodel}, we discuss alternative error models and how they relate to our chosen error model.

\begin{definition}[$(\varepsilon,\delta)$-far]\label{def:far}
Given two curves $\firstcurve$ and $\secondcurve$ consisting of $n$ vertices each, we say that $\firstcurve$ and~$\secondcurve$ are $(\varepsilon,\delta)$-far from each other if there exists no coupling path from $(1,1)$ to $(n,n)$ in the $\delta$-free space matrix of cost $\varepsilon n$ or less.
\end{definition}

\begin{definition}[Fréchet-tester]
Assume we are given query-access to two curves $\firstcurve$ and $\secondcurve$, and we are given values $\delta > 0$ and $0<\eps<2$. If the two curves have Fréchet distance at most $\delta$, we must return \enquote{yes},
and if they are $(\eps,\delta)$-far from each other w.r.t.\ the Fréchet distance, the algorithm must return \enquote{no}, with probability at least $\succprob$. 
\end{definition}

\subsection{A failed attempt}
A natural way to attack the problem is to cast it as function monotonicity testing in the free space matrix (see also Section 1.6 in \cite{bhattacharyya2022property}). Here, for each column $i$ one has to choose the row $f(i)$ where the coupling path goes through. The challenge here is that there may be $\mathcal{O}(\tloc)$ possible values. The classical function monotonicity testing algorithm works by performing a binary search for the value of $f(i)$ in the array $[f(1),...f(n)]$, and declaring the value to be \enquote{good} if the search ends in position $i$. 
To emulate this approach, one could choose a sleeve of height $\mathcal{O}(\tloc)$ around the zero-entries in the columns of the free space matrix. Thus, testing if a column $i$ is good, now becomes a problem of checking if a coupling path exists through this sleeve in the submatrix corresponding to a small number (i.e. $\mathcal{O}(\tloc\log\n)$) of columns around $i$.
The problem of finding a coupling path is easy to solve. Unfortunately, extending the proof of correctness of the monotone function testing to this case turns out to be surprisingly challenging, and in any case, this would lead to a result that is inferior to the one we present in this paper.

\subsection{Main results}

We present two algorithms for testing whether a pair of input curves have discrete Fr\'echet distance at most a given real value~$\delta$. Assume that the algorithm has query access only to the $\delta$-free space matrix of the input curves and that this matrix is $t$-local (see Definition~\ref{definition:t-local}).
The first algorithm requires knowledge of $t$, and it uses $O(\frac{t}{\eps} \log\frac{t}{\eps})$ queries (Theorem~\ref{thm:frechet-tester-known-t}). Note that this bound is intrinsically independent of $n$ while  $t$ can be at most $n$, by the definition of locality. The second algorithm does not require knowledge of~$t$ in advance (and thus it can be applied to any matrix) and requires
$\mathcal{O}((\tloc^3+\tloc^2\log\n)\frac{\log\log\tloc}{\eps})$
many queries (Theorem~\ref{thm:frechet-tester-unknown-t}).
Using Lemma~\ref{lem:straight} these results directly imply the same bounds for the discrete Fréchet distance of $\tloc$-straight curves with bounded edge lengths. In Section~\ref{sec:ext} we relax the bounded edge length assumption towards a uniform sampling condition (Theorem~\ref{thm:frechet-tester-edge-lengts}).
In Section~\ref{sec:continuous} we show how to apply our algorithms to testing the continuous Fréchet distance. As a result, we obtain essentially the same bounds with respect to $t$-straight curves  while removing the assumption on the edge lengths entirely (Theorem~\ref{thm:frechet-tester-continuous} and Corollary~\ref{cor:final}). Note that for the second Fréchet tester, we assume an aspect ratio that is polynomial in $n$. 

\subsection{Techniques}

Our algorithms use the concept of permeability, which tests whether a specified subcurve has small partial Fréchet distance to any subcurve of the other curve. Concretely, the algorithm samples a block of consecutive columns (or rows) in the $\delta$-free space matrix and checks if there is a coupling path of cost $0$ passing somewhere through this block (see Algorithm~\ref{algo:frechet-tester-exact}).
If the algorithm happens to find a subcurve that does not satisfy this condition, then we know that the global Fréchet distance is larger than the distance parameter $\delta$ and the algorithm can safely return 'no'.
The difficulty lies in proving that after a certain number of random queries, the algorithm can return 'yes' with sufficient certainty. To this end, we show that if the two curves are $(\eps,\delta)$-far under the (discrete) Fréchet distance then the algorithm, which tests subcurves of varying sizes up to roughly $O(\frac{t}{\eps})$, is likely to sample a permeability query that would return 'no'. 

Note that our first algorithm crucially needs to know the value of locality $t$ to guide the process of selecting the subcurves to test. The second algorithm needs to gain this knowledge of $t$ through its queries, since the value of $t$ is not known in advance. For this, we propose  a variant of monotonicity testing (Algorithm~\ref{algo:t-locality-tester}). Combining the output of the locality tester with the first algorithm turns out to be technically challenging, since the output of a property tester always comes with the caveat of unknown input modifications, but using a careful adaptation we manage to utilize similar arguments as in the previous setting (see Algorithm~\ref{algo:combined-frechet-tester}).

Section~\ref{sec:additional} contains further additional results, such as a simple tester for the Hausdorff distance and a simple $t$-approximate Fréchet tester (Algorithm~\ref{algo:hausdorff-tester}). This algorithm  merely tests for columns and rows that contain one-entries only and are therefore somewhat blocked for coupling paths in the free space diagram. These are essentially range-emptiness queries centered on individual vertices. 
The example in Figure~\ref{fig:example-permeability-queries-necessary} shows that such simple queries are not sufficient to test the Fréchet distance exactly.
In the example, all rows and columns are individually permeable, but the Fréchet distance is large. This is why Algorithm~\ref{algo:frechet-tester-exact} tests for permeability of blocks of varying sizes. 

\begin{figure}[h]
    \centering
    \includegraphics[width=0.275\linewidth, page=1]{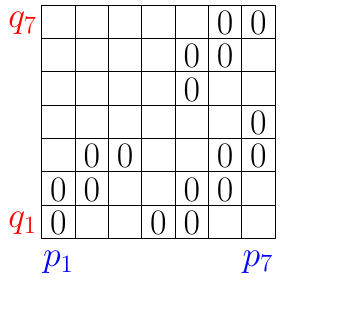}
    \hspace{4em}
    \includegraphics[width=0.275\linewidth, page=2]{example-need-perm-queries.pdf}
    \caption{On the left we have $\fsm_{\delta}$ for the curves on the right. All zero-entries are written. 
    There is no coupling path of cost $0$ in $\fsm_{\delta}$ but there is also no barrier-column or barrier-row.}
    \label{fig:example-permeability-queries-necessary}
\end{figure}

\subsection{Alternative error models}\label{sec:errormodel}

Definition~\ref{def:far} allows the path to flip $\varepsilon n$ one-entries along the coupling path, but we could instead permit other vertex operations (such as vertex deletions) on the curves and thereby change the free space matrix to locally `repair' the coupling path. 
We argue below that our error model is not significantly more powerful than allowing any combination of these vertex operations. 
Concretely, consider the following operations.
\begin{enumerate}
    \item Modification: Modify the coordinates of a vertex of one of the curves. 
    \item Deletion: Delete a vertex from one of the curves.
    \item Insertion: Insert a new vertex into one of the curves at a given position.
\end{enumerate}
Consider a sequence of $k$ one-entries on the optimal coupling path that need to be flipped.
Let $(i,j)$ be the last index before and let $(i',j')$ be the first index after the flipped sequence along the path. We delete the rows with index $i+1$ up to $i'-1$ and the columns with index $j+1$ up to $j'-1$. This uses at most $2k$ deletions and has the same effect that $(i',j')$ is reachable from $(i,j)$ with a coupling path (by using a diagonal step).
Similarly, note that we can simulate vertex deletions using vertex modifications by assigning to a vertex or a group of vertices the coordinates of a neighboring vertex that is not deleted. This has the same effect on the Fr\'echet distance of the input curves as deleting this group of vertices.
Lastly, we can also simulate vertex deletions by inserting a copy of the respective vertex in the other curve at the corresponding position along the coupling, so that the deleted vertex can be matched to the newly inserted vertex. 
Thus, the number of edits made using vertex operations will be the same or less, up to constant factors, when compared to the minimum cost of a coupling path. 

Moreover, we remark that these operations might `destroy' locality of the free space matrix and some of them change the indexing of entries. For these reasons, we prefer to use the cleaner variant of direct edits (flips) in the free space matrix provided by Definition~\ref{def:far}.

\section{Testing the discrete Fréchet distance}\label{sec:frechet-with-known-t}

In this section, we describe and analyze a Fréchet tester under the assumption that the free space matrix is $t$-local for a given value of $t$. 

\subsection{The algorithm}

The idea of Algorithm~\ref{algo:frechet-tester-exact} is to sample a number of columns and rows and check whether there is locally a coupling path of cost zero possible. The following definition classifies when such a (local) path exists.
\begin{definition}[Permeability] We say a block $[i,i']$ of consecutive columns (resp., rows) from index $i$ to index $i'$ is \emph{permeable} if there exists a coupling path of cost zero that starts in column (resp., row) $i$ and ends in column (resp., row) $i'$. 
\end{definition}
If a column (resp. row) is individually not permeable, i.e. it  contains only one-entries, we call it a \emph{barrier-column} (resp. \emph{barrier-row}).
Note that any non-permeable block of rows or columns is a witness to the fact that no global coupling path of cost $0$ exists and that the algorithm can answer \enquote{no}.

The algorithm first tests if the first or last entry (i.e.\ $M[1,1]$ or $M[\n,\n]$) is a one-entry. If not, it queries $\mathcal{O}\left(\frac{\tloc}{\eps}\right)$ randomly sampled columns and rows and checks if any of them is a barrier-column or barrier-row. In Lines~\ref{line:interval-samp-beg}-\ref{line:interval-samp-end} we sample $\mathcal{O}(\frac{t}{2^{i+1}})$ random intervals of length $2^{i+2}$ for each $i$ from $0$ to $\mathcal{O}\left(\log\frac{\tloc}{\eps}\right)$. We then check all corresponding blocks of columns and rows for permeability and if all blocks are permeable, we return \enquote{yes}.

\begin{algorithm}[h]
 \caption{Fréchet-Tester1($M,\tloc,\eps,k=\frac{24t}{\eps},c=4$)}
    \label{algo:frechet-tester-exact}
\begin{enumerate}
\item \textbf{If} $M[1,1]=1$ or $M[\n,\n]=1$ \textbf{then} \textbf{return} \enquote{no}.
\item \textbf{repeat} $\lceil k\rceil$ \textbf{times}:
\item ~~ $j \gets$ sample an index uniformly at random from $[\n]$.
\item ~~ \textbf{if} row $j$ or column $j$ of $\fsm$ is a barrier-column or barrier-row \textbf{then} \textbf{return} \enquote{no}.\label{line:found-zero-col-row}
\item $K\gets\lceil\frac{\eps\n}{32\tloc}\rceil-1$, $\gets\ell=\lceil\frac{128\tloc}{\eps}\rceil$, let $\J$ be a set of intervals and set $\J\gets\emptyset$.
\item \textbf{for} {$i=0,\ldots,\lfloor\log_2\ell\rfloor$} \textbf{do}:\;\label{line:interval-samp-beg}
    {
\item ~~ $I\gets$ sample $\lceil \frac{4c\n}{2^{i+1}K} \rceil$ different indices uniformly at random from $\{0,1,\ldots,\frac{n}{2^{i+1}}-2\}$.
\item ~~ \textbf{for each} $j\in I$ \textbf{do}: add $[j2^{i+1},(j+2)2^{i+1}]$ to $\J$.\label{line:interval-samp-end}
    }
\item \textbf{foreach} $[i,j]\in\J$ \textbf{do}
\item ~~\textbf{if} block $[i,j]$ of consecutive columns is not permeable \textbf{then} \textbf{return} \enquote{no}.\label{line:permeability-query-col}
\item ~~\textbf{if} block $[i,j]$ of consecutive rows is not permeable \textbf{then} \textbf{return} \enquote{no}.\label{line:permeability-query-row}
\item \textbf{return} \enquote{yes}.
\end{enumerate}
\end{algorithm}

To check if a block of $k$ columns (or rows) is permeable, the algorithm first performs $k$ queries to obtain the positions of zero-entries in these columns (or rows), then we build the induced subgraph of the grid for these zero-entries, connect all neighboring zero-entries according to the possible steps of a coupling and then connect all zero-entries of the last column to a sink and all zero-entries of the first column from a source. It remains to check if there is a path from source to sink, which can be done in linear time since the graph is acyclic. The running time is linear in the total number of zero-entries queried.
We call this a \emph{permeability query}.

\subsection{Basic properties of permeability}
For two curves $\firstcurve$ and $\secondcurve$ with $\n$ vertices each, and a $\tloc$-local free space matrix, Algorithm~\ref{algo:frechet-tester-exact} returns \enquote{no} only if it has determined that no coupling path of cost zero from $(1, 1)$ to $(\n,\n)$ can exist (by finding blocks of rows or columns that are not permeable). We need to show that it returns \enquote{yes} correctly with sufficiently high probability. 

For technical reasons, our analysis uses a specific definition of coupling paths that deviates from the one given in Section~\ref{sec:prelims}. We introduce it here and we argue that it suffices to consider these coupling paths only.

\begin{definition}[\restricteddiagonals\ path]\label{def:rdpath}
    A \emph{\restricteddiagonals\ path} $\optcoupling$ through the free space matrix $\fsm$ is a path corresponding to a coupling $\optcoupling$ in which for any consecutive tuples of the form $(i,j),(i+1,j+1)\in\optcoupling$ it holds that $\fsm[i,j]=0=\fsm[i+1,j+1]$.
\end{definition}

In other words, a \restricteddiagonals\ path $\optcoupling$, is allowed to take diagonal steps in the free space matrix only if both entries are zero. Hence, any one-entry visited by $\optcoupling$ must be reached by either a vertical or horizontal step in $\fsm$ and it must be left by a vertical or horizontal step in $\fsm$. This definition, which might seem overly technical at first sight, will simplify a lot of our arguments.
The following observations show that it suffices to consider \restricteddiagonals\ paths instead of all coupling paths in the analysis of a Fréchet tester.

\begin{observation}
    For two curves $P$ and $Q$ with $\n$ vertices each, the following holds:
    \begin{itemize}
        \item [i)] There exists a \restricteddiagonals\ path of cost zero from $\fsm_{\delta}[1,1]$ to $\fsm_{\delta}[\n,\n]$ if and only if there exists a coupling path of cost zero  from $\fsm_{\delta}[1,1]$ to $\fsm_{\delta}[\n,\n]$.
        \item[ii)] If there is no \restricteddiagonals\ path of cost at most $\eps\n$ through $\fsm_{\delta}$ from $\fsm_{\delta}[1,1]$ to $\fsm_{\delta}[\n,\n]$, the curves are $(\frac{\eps}{3},\delta)$-far. See Figure~\ref{fig:restricted-diagonals}.
        \item[iii)] If $P$ and $Q$ are $(\eps,\delta)$-far, there is no \restricteddiagonals\ path of cost at most $\eps\n$ through $\fsm_{\delta}$ from $\fsm_{\delta}[1,1]$ to $\fsm_{\delta}[\n,\n]$.
    \end{itemize}
\end{observation}

\begin{figure}
    \centering
    \includegraphics[width=0.2\linewidth]{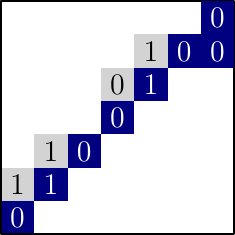}
    \caption{The blue path is a coupling path of cost $2$. By adding the gray entries, we get a \restricteddiagonals\ path of cost $5$.}
    \label{fig:restricted-diagonals}
\end{figure}
Below, we present three structural lemmas that help with the analysis of Algorithm~\ref{algo:frechet-tester-exact}. The first lemma considers a subpath of a \restricteddiagonals\ path $\optcoupling$ that starts and ends in a zero-entry and visits only one-entries in between. We show that the bounding box of the two zero-entries must otherwise be filled with one-entries. Note that this lemma does not hold for general coupling paths of Section~\ref{sec:prelims}. Using it will simplify our analysis and is the main motivation for introducing Definition~\ref{def:rdpath}.

\begin{lemma}[Box of ones]\label{lem:box-of-zeros}
    Let $\optcoupling$ be a \restricteddiagonals\ path of minimum cost. Let $(i,j),(i',j')\in\optcoupling$ be zero-entries such that $\optcoupling$  visits only one-entries between them. Then, for all tuples $(k,l) \in \{(k,l)\mid  i\leq k \leq i', j\leq l\leq j' \} \setminus \{(i,j), (i',j')\}$
    we have that $\fsm[k,l]=1$.
\end{lemma}
\begin{proof}
    Since $\optcoupling$ is a \restricteddiagonals\ path, it visits exactly $i'-i+j'-j-1$ one-entries.
    For any entry $(k,l)\in \{(k,l)\mid  i\leq k \leq i', j\leq l\leq j' \} \setminus \{(i,j), (i',j')\}$, there exists a coupling path $\pi$ from $(i,j)$ to $(i',j')$ that passes through $(k,l)$ and makes no diagonal steps. If $\fsm[k,l]=0$, $\pi$ visits at most $i'-i+j'-j-2$ one-entries and we can replace the subpath of $\optcoupling$ from $(i,j)$ to $(i',j')$ by $\pi$, which yields a contradiction to the minimality of $\optcoupling$.
\end{proof}

    \begin{figure}[t]
        \centering
        \includegraphics[width=0.5\linewidth]{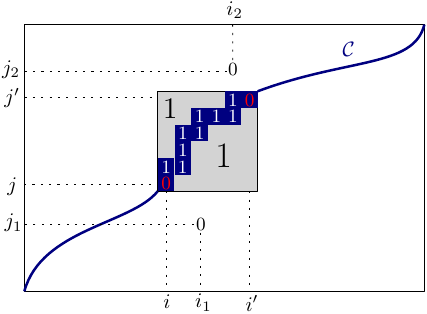}
        \caption{Box of one-entries. Illustration to the proofs of Lemma~\ref{lem:box-of-zeros} and the Barrier Lemma.}
        \label{fig:zero-values-in-deleted-cols-FSM}
    \end{figure}
    
We introduce a definition that describes when two columns or rows satisfy the constraints for $\tloc$-locality.

\begin{definition}\label{def:pass-tloc}
    The columns $i_1$ and $i_2$ \emph{pass $\tloc$-locality} if any two zero-entries $(i_1,j_1)$ and $(i_2,j_2)$ in columns $i_1$ and $i_2$ satisfy $\abs{j_1-j_2}\leq\tloc\cdot(2+\abs{i_1-i_2})$. 
    Similarly the rows $j_1$ and $j_2$ \emph{pass $\tloc$-locality} if any two zero-entries $(i_1,j_1)$ and $(i_2,j_2)$ in rows $j_1$ and $j_2$ satisfy $\abs{i_1-i_2}\leq\tloc\cdot(2+\abs{j_1-j_2})$.
\end{definition}
A column (resp. row) can also pass $\tloc$-locality with itself if $i_1=i_2$ (resp. $j_1=j_2$). 
Intuitively, two columns pass $\tloc$-locality, when all their zero-entries are not too far away vertically.
Note that all columns (resp.\ rows) pass $\tloc$-locality with each other if $\fsm$ is $\tloc$-local.

The second lemma shows that if an optimal \restricteddiagonals\ path contains a long sequence of one-entries, then there must be many barrier-columns and barrier-rows.

\begin{lemma}[Barrier Lemma]\label{lem:barrier-lemma}
    Let $\optcoupling$ be a \restricteddiagonals\ path of minimum cost. Suppose that there are two zero-entries $(i,j),(i',j')\in\optcoupling$ such that $\optcoupling$ visits no zero-entry and at least $4\tloc$ one-entries in between them 
    and suppose that all columns in $[i,i']$ and all rows in $[j,j']$ pass $\tloc$-locality with each other. 
    Then there is a total of at least $\lceil\frac{i'-i+j'-j}{2\tloc}\rceil$ barrier-rows between $j$ and $j'$ and barrier-columns between $i$ and $i'$.
\end{lemma}

\begin{proof}
    Let $i_1\leq i'$ be maximal such that $\fsm[i_1,j_1]=0$ for some $j_1\leq j$ and let $i_2\geq i_1$ be minimal such that $\fsm[i_2,j_2]=0$ for some $j_2\geq j'$. Hence, $(i_1,j_1)$ is the rightmost zero-entry \enquote{below} the box of ones from Lemma~\ref{lem:box-of-zeros} and $(i_2,j_2)$ is the leftmost zero-entry to the right of $i_1$ that is \enquote{above} the box of ones. Note that $(i,j)$ is a candidate for $(i_1,j_1)$ and $(i',j')$ is a candidate for $(i_2,j_2)$. Refer to Figure~\ref{fig:zero-values-in-deleted-cols-FSM}. By definition, all columns between $i_1$ and $i_2$ are barrier-columns. We can use that $i_1$ and $i_2$ pass $\tloc$-locality on $(i_1,j_1)$ and $(i_2,j_2)$ to get 
    \[i_2-i_1\geq\frac{j_2-j_1}{\tloc}-2\geq\frac{j'-j}{\tloc}-2.\]
    So there are at least $\frac{j'-j}{\tloc}-2$ barrier-columns between $i$ and $i'$. Analogously, we can prove that there are at least $\frac{i'-i}{\tloc}-2$ barrier-rows between $j$ and $j'$. Thus, in total we get that the number of barrier-columns and barrier-rows in the stated range is at least
    \[\frac{i'-i+j'-j}{\tloc}-4=\frac{i'-i+j'-j}{2\tloc}+\frac{i'-i+j'-j-4\tloc}{2\tloc}\geq\frac{i'-i+j'-j}{2\tloc}.\]
    The last inequality follows since in order to visit at least $4\tloc$ one-entries between $(i,j)$ and $(i',j')$, it must hold that $i'-i+j'-j\geq4\tloc$. The statement of the lemma now follows since the number of rows and columns is an integer.
\end{proof}

The next lemma shows that if an optimal \restricteddiagonals\ path has a long stretch with relatively many one-entries on the path, there cannot be any long coupling path of cost zero in the same sequence of rows or columns of the matrix $\fsm$, implying impermeability.

\begin{lemma}[Impermeability Lemma]\label{lem:impermeability-lemma}
    Let $\optcoupling$ be a \restricteddiagonals\ path through $\fsm$ of minimum cost. Suppose $(i,j),(i',j')\in\optcoupling$ with $j'-j>2\tloc$ (resp. $i'-i>2\tloc$) correspond to zero-entries in $\fsm$ and the subpath of $\optcoupling$ from $(i,j)$ to $(i',j')$ visits at least $4\tloc-1$ one-entries
    and suppose that any column in $[i,i']$ and any row in $[j,j']$ passes $\tloc$-locality with itself. 
    Then, the block of columns $[i,i']$ (resp. rows $[j,j']$) is not permeable.
\end{lemma}

\begin{proof}
    We prove the statement for $j'-j>2\tloc$. The statement for $i'-i>2\tloc$ can be proven analogously.
    For the sake of contradiction, assume that there exists a coupling path $\pi$ of cost zero from $(i,j_1)$ to $(i',j_2)$ in $\fsm$. 
    We perform a case analysis based on the positions of $j_1$ and $j_2$ with respect to $j$ and $j'$. An illustration of these cases can be seen in Figure~\ref{fig:cases-for-xy-monotone-path-beginning-and-end}. Our proof strategy in each case is to locally modify the path $\optcoupling$ to show a contradiction to its minimality. These modifications are also illustrated in Figure~\ref{fig:cases-for-xy-monotone-path-beginning-and-end}.

    \begin{figure}[t]
       \centering
       \includegraphics[width=0.22\linewidth, page=2]{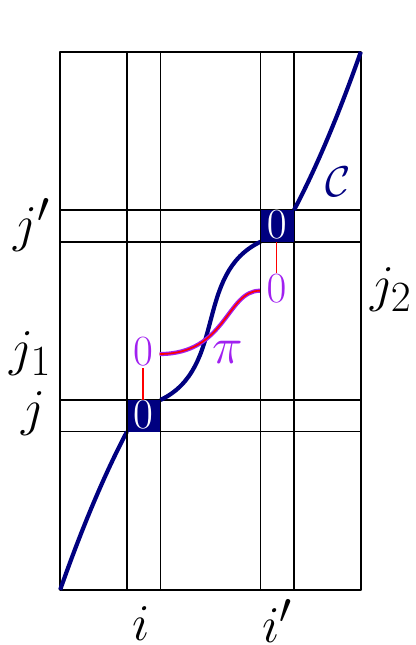}
       \hspace{1em}
       \includegraphics[width=0.22\linewidth, page=3]{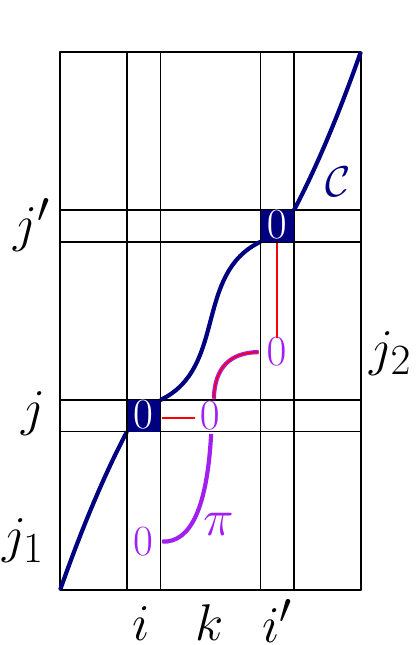}
       \hspace{1em}
       \includegraphics[width=0.22\linewidth, page=5]{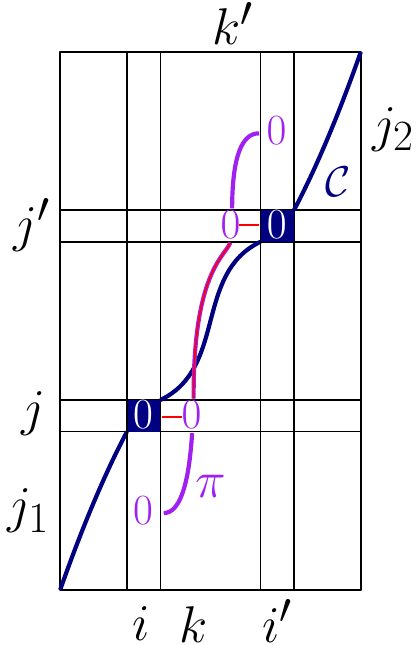}
       \hspace{1em}
       \includegraphics[width=0.22\linewidth, page=4]{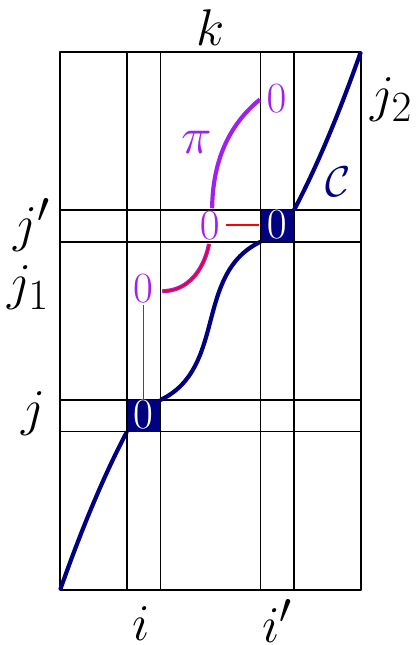}
       \caption{The cases in the Impermeability Lemma in left to right order. The red lines show the modifications to possibly improve $\optcoupling$.}
       \label{fig:cases-for-xy-monotone-path-beginning-and-end}
    \end{figure}
    \begin{enumerate}

    \item  \label{case:left-one-above-right-one-below} \textbf{Case \ref{case:left-one-above-right-one-below}:} ($j_1\geq j$ and $j_2<j'$)
 
    Since $i$ and $i'$ pass $\tloc$-locality with themselves, we know that $j_1-j\leq2\tloc$ and $j'-j_2\leq2\tloc$. We modify the path $\optcoupling$ as follows. Instead of going from $(i,j)$ to $(i,j')$ via $\optcoupling$, we first go from $(i,j)$ to $(i,j_1)$, then we follow $\pi$ until $(i',j_2)$ and then go from $(i',j_2)$ to $(i',j')$. This visits at most $4\tloc-2$ many one-entries instead of locally visiting at least $4\tloc-1$ many one-entries.

    \item \label{case:both-ones-below} \textbf{Case \ref{case:both-ones-below}:} ($j_1<j$ and $j_2<j'$) 

    First, by the fact that $i'$ passes $\tloc$-locality with itself and the setting of the lemma, we know that $j'-j_2\leq2\tloc<j'-j$ and therefore, $j_2>j$. So the path $\pi$ passes row $j$. Let $k$ be such that $(k,j)\in\pi$. Observe that $k-i\leq2\tloc$ since $j$ passes $\tloc$-locality with itself. 
    We modify the path $\optcoupling$ as follows. Instead of going from $(i,j)$ to $(i,j')$ via $\optcoupling$, we first go from $(i,j)$ to $(k,j)$, then we follow $\pi$ until $(i',j_2)$ and then go from $(i',j_2)$ to $(i',j')$. This visits at most $4\tloc-2$ many one-entries instead of locally visiting at least $4\tloc-1$ many one-entries.
    
    \item \label{case:left-one-below-right-one-above} \textbf{Case \ref{case:left-one-below-right-one-above}:} ($j_1<j$ and $j_2\geq j'$)

    In this case, the path $\pi$ passes both row $j$ and row $j'$. Let $k$ and $k'$ be such that $(k,j)\in\pi$ and $(k',j')\in\pi$. Since $j$ and $j'$ pass $\tloc$-locality with themself, we have $k-i\leq2\tloc$ and $i'-k'\leq2\tloc$. 
    We modify the path $\optcoupling$ as follows. Instead of going from $(i,j)$ to $(i,j')$ via $\optcoupling$, we first go from $(i,j)$ to $(k,j)$, then we follow $\pi$ until $(k',j')$ and then go from $(k',j')$ to $(i',j')$. This visits at most $4\tloc-2$ many one-entries instead of locally visiting at least $4\tloc-1$ many one-entries.
        
    \item \label{case:both-ones-above} \textbf{Case \ref{case:both-ones-above}:}  ($j_1\geq j$ and $j_2\geq j'$)

    Case~\ref{case:both-ones-above} is symmetric to Case~\ref{case:both-ones-below}  by reversing the direction of both $\firstcurve$ and $\secondcurve$.\qedhere
    \end{enumerate}
\end{proof}

\subsection{Analysis of the algorithm}

Algorithm~\ref{algo:frechet-tester-exact} samples a range of  intervals in Lines~\ref{line:interval-samp-beg}-\ref{line:interval-samp-end} which are then tested with a permeability query. For this subroutine we can prove the following lemma.

\begin{lemma}\label{lem:interval-sampling}
    Let $\I = \left\{ [i,j]\, |\, 1 \le i < j \le \n\right\}$ be a set of all non-empty intervals from $[\numvertfirst]$ and let $X\subset \I$ be a set of $K$ unknown intervals such that the length of each interval is at most $\ell$ and each element of $[\n]$ is contained in at most $c$ of these intervals for some constant $c\geq1$. 
    Lines~\ref{line:interval-samp-beg}-\ref{line:interval-samp-end} of Algorithm~\ref{algo:frechet-tester-exact} select a subset $\J \subset \I$ of intervals with 
    $ \sum_{J \in \J} |J| =  O(\frac{n\log \ell}{K})$ such that 
    with probability at least $9/10$, at least one of the intervals in $X$ is contained
    in at least one interval in $\J$.
\end{lemma}
\begin{proof}
    Let $X_i \subset X$ be the subset of intervals that have length between $2^{i}$ and $2^{i+1}$
    and let $K_i = |X_i|$.
    Consider the set $C_i$ of \emph{candidate intervals} $[j2^{i+1},(j+2)2^{i+1}]$ for every  
    $j=0, \ldots, \frac{\numvertfirst}{2^{i+1}}-2$. Algorithm~\ref{algo:frechet-tester-exact} samples $\frac{4c\numvertfirst}{2^{i+1}K}$ of these candidate intervals in Lines~\ref{line:interval-samp-beg}-\ref{line:interval-samp-end} uniformly at random.
    Observe that for a fixed value of $i$, the total length of the sampled intervals is  
    $\mathcal{O}(\numvertfirst/K)$ leading to the claimed total size.
    
    Thus, it remains to analyze the correctness.  
    Observe that each interval in $X_i$ is contained in at least one candidate interval in $C_i$ and $\abs{C_i}\leq\frac{\numvertfirst}{2^{i+1}}$.
    If we sample one interval $J$ from $C_i$, the probability that $J$ does not contain any of the intervals from $X_i$ is at most $1-\frac{2^{i+1}K_i}{c\numvertfirst}$. The additional factor of $c$ in the denominator comes from the fact that each element of $[\n]$ can be contained in up to $c$ different intervals. Hence, the probability that we do not sample any interval that contains any of the intervals from $X_i$ is at most 
    \[\left(1-\frac{2^{i+1}K_i}{c\numvertfirst}\right)^{4c\numvertfirst/(2^{i+1}K)}\leq e^{-4K_i/K}.\]
    As the iterations are independent, the probability of failing in all of the $\lfloor\log_2\ell\rfloor+1$ steps is at most
    \[\Pi_{i=0}^{\lfloor\log_2\ell\rfloor}e^{-4K_i/K}=e^{-4}<\frac{1}{10}.\qedhere\]
\end{proof}

The next lemma shows that Algorithm~\ref{algo:frechet-tester-exact} is very likely to return \enquote{no} in Line~\ref{line:found-zero-col-row} if there are many barrier-columns and barrier-rows.

\begin{lemma}\label{lem:many-empty-rows-and-cols}
    If the total number of barrier-columns and barrier-rows is more than $\frac{3\n}{k}$, then Algorithm~\ref{algo:frechet-tester-exact} returns \enquote{no} in Line~\ref{line:found-zero-col-row} with probability at least $\frac{9}{10}$.
\end{lemma}

\begin{proof}
    Suppose there are $z_r$ barrier-rows and $z_c$ barrier-columns with $z_r+z_c>\frac{3\n}{k}$. If we sample $j\in[\n]$ uniformly at random, we have that row $j$ is a barrier-row with probability $\frac{z_r}{\n}$ and column $j$ is a barrier-column with probability $\frac{z_c}{\n}$. So after $\lceil k\rceil$ iterations, the probability that none of the sampled columns or rows is a barrier-column or barrier-row is at most
    \[\left(1-\frac{z_r}{\n}\right)^{k}\left(1-\frac{z_c}{\n}\right)^{k}\leq e^{-(z_r+z_c)k/\n}\leq e^{-3}<\frac{1}{10}.\]
    So the algorithm returns \enquote{no} with probability at least $\frac{9}{10}$.
\end{proof}

\begin{figure}
    \centering
    \includegraphics[width=0.45\linewidth]{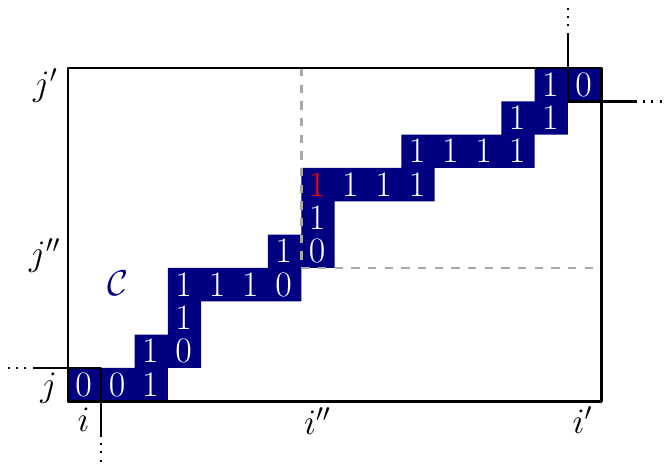}
    \caption{Example for grouping in the proof of Lemma~\ref{lem:main-lemma} with $\tloc=2$. The group starts in $(i,j)$. The red one-entry is the $(4\tloc+1)$-th one-entry in the group. So we stop the group at the next zero-entry $(i',j')$. Since it contains more than $8\tloc$ one-entries, this is a hollow group. The entries $(i'',j''),(i',j')$ satisfy the conditions in the Barrier (Lemma~\ref{lem:barrier-lemma}).}
    \label{fig:grouping-first}
\end{figure}

\begin{lemma}[Main Lemma]\label{lem:main-lemma}
    Let $\fsm$ be $\tloc$-local and $\fsm[1,1]=0=\fsm[\n,\n]$. Suppose the total number of barrier-columns and barrier-rows is at most $\frac{\eps\n}{8\tloc}$. Let $\optcoupling$ be a \restricteddiagonals\ path with lowest cost through $\fsm$ and suppose that $c(\optcoupling)>\eps\n$. 
    Then, with probability at least $\frac{9}{10}$ at least one sampled interval in the set $\J$ during the execution of Algorithm~\ref{algo:frechet-tester-exact} corresponds to a non-permeable block of columns or rows. 
\end{lemma}
\begin{proof}
    We first divide the path $\optcoupling$ into a minimum number of groups that each contain more than $4\tloc$ one-entries and start and end with a zero-entry. A visualization of this can be seen in Figure~\ref{fig:grouping-first}. Here, the start of one group is the same entry as the end of the prior group. We do this greedily by following $\optcoupling$ until we have visited $4\tloc+1$ one-entries. Then, we end the group at the next zero-entry after that and repeat the process. Note that the last group may contain fewer than $4\tloc+1$ one-entries.

    We will first show that there are not too many groups that contain more than $8\tloc$ one-entries. Then, we will show that within the remaining groups, there are sufficiently many groups that are not too big. With these groups we then form the set $X$ for the interval sampling in Lemma~\ref{lem:interval-sampling} (i.e., each interval in $X$ will induce a permeability query that fails).

    A group that contains more than $8\tloc$ one-entries is called \textit{hollow}, otherwise, it is \textit{dense}.
    Let $H_1,\ldots,H_k$ be the hollow groups with $x_1,\ldots,x_k$ one-entries.
    \begin{claim}\label{claim:few-hollow-groups}
        At most $\frac{\eps\n}{2}$ of the one-entries of $\optcoupling$ are in hollow groups.
    \end{claim}
    \begin{proof}[Proof of Claim~\ref{claim:few-hollow-groups}]
        Let $H$ be one of the hollow groups and let $x$ be the number of one-entries in $H$.        
        Let $(i',j')$ and $(i'',j'')$ be the last and second to last zero-entries of $H$. By the greedy construction of $H$ we know that at least $x-4\tloc$ of the one-entries appear between $(i'',j'')$ and $(i',j')$.
        So we have that $i'-i''+j'-j''\geq x-4\tloc>4\tloc$ and therefore, $(i'',j'')$ and $(i',j')$ satisfy the conditions of the Barrier Lemma (Lemma~\ref{lem:barrier-lemma}). Hence, the number of barrier-columns in $[i'',i']$ and barrier-rows in $[j'',j]$ is at least
        \[\frac{i'-i''+j'-j''}{2\tloc}\geq\frac{x-4\tloc}{2\tloc}\geq\frac{x}{4\tloc}.\]
        The last inequality follows from the fact that $x>8\tloc$ and therefore $x-4\tloc\geq\frac{x}{2}$. 
        Summing up these inequalities for all groups $H_1,\ldots,H_k$, we get that the total number of barrier-columns and barrier-rows is at least $\frac{1}{4\tloc}\sum_{i=1}^kx_i$. By assumption, we also have that 
        the total number of such columns and rows is at most $\frac{\eps\n}{8\tloc}$ and therefore, $\sum_{i=1}^kx_i\leq\frac{\eps\n}{2}$.
    \end{proof}
    So we know that at least $\frac{\eps\n}{2}$ of the one-entries of $\optcoupling$ are in dense groups. Since a dense group contains at most $8\tloc$ one-entries, we have at least $\lceil\frac{\eps\n}{16\tloc}\rceil$ dense groups.
    Let $G$ be a dense group that starts in $(i,j)$ and ends in $(i',j')$. We say that $i'-i+j'-j+2$ is the \emph{size} of $G$. This corresponds to the sum of the number of columns and rows that it visits. In total, every row and column is counted at most twice by the definition of the groups. So the sum of all group sizes is at most $4\n$. 
    We observe that at least half of the groups (which are at least $\lceil\frac{\eps\n}{32\tloc}\rceil$ many) have size at most $2\cdot4\n/\lceil\frac{\eps\n}{16\tloc}\rceil\leq\frac{128\tloc}{\eps}$. We call these groups \emph{small}.
    At last, we create a set $X$ of intervals so that we can apply Lemma~\ref{lem:interval-sampling} to it. Each of the intervals $[i,i']$ in $X$ will have the property that either the block $[i,i']$ of consecutive columns or the block $[i,i']$ of consecutive rows is not permeable.
    
    Let $G$ be a small dense group and suppose it starts in $(i,j)$ and ends in $(i',j')$.
    We observe that unless $G$ is the very last group, it visits at least $4\tloc+1$ one-entries and therefore either visits more than $2\tloc$ rows or more than $2\tloc$ columns. In the first case, we can apply the Impermeability Lemma to see that 
    the block $[i,i']$ of columns is not permeable. So we add $[i,i']$ to $X$. In the second case, we can apply the Impermeability Lemma to see that 
    the block $[j,j']$ of rows is not permeable. So we add $[j,j']$ to $X$. In the end the set $X$ has at least $\lceil\frac{\eps\n}{32\tloc}\rceil-1$ elements and each of them has length at most $\frac{128\tloc}{\eps}$. 
    If we take a closer look at the definitions of the groups, we see that all columns are only contained in one group except for the first and last column of each group. So we can say that every column is contained in at most two groups, i.e. at most two intervals of $X$. The same holds for rows. So, we have that each element of $[\n]$ is contained in at most four of the intervals in $X$. 
    So we can apply Lemma~\ref{lem:interval-sampling} to $X$ with $\ell=\lceil\frac{128\tloc}{\eps}\rceil$, $K=\lceil\frac{\eps\n}{32\tloc}\rceil-1$, $c=4$ to see that the algorithm samples one of the intervals from $X$ with probability at least $\frac{9}{10}$.
    So the permeability query fails for the columns or for the rows
    with the respective interval fails and we thus have a failing permeability query with probability at least $\frac{9}{10}$.
\end{proof}

\begin{lemma}\label{lem:number-queries}
    Algorithm~\ref{algo:frechet-tester-exact} performs $\mathcal{O}(\frac{\tloc}{\eps}\log\frac{\tloc}{\eps})$ queries.
\end{lemma}
\begin{proof}
    Line~\ref{line:found-zero-col-row} performs $\mathcal{O}(\frac{\tloc}{\eps})$ queries. By Lemma~\ref{lem:interval-sampling}, Lines~\ref{line:interval-samp-beg}-\ref{line:interval-samp-end} sample a set $\J$ of intervals that contain a total of $\mathcal{O}(\frac{\n\log\ell}{K})$ rows and columns. Each of these columns and rows cause a query in a permeability query. As $\ell=\lceil\frac{128\tloc}{\eps}\rceil$ and $K=\lceil\frac{\eps\n}{32\tloc}\rceil-1$, the claim follows.
\end{proof}

\begin{theorem}
    \label{thm:frechet-tester-known-t}
    Let $\delta>0$ and $0<\eps<2$ be given. Let $\firstcurve$ and $\secondcurve$ be curves with $\n$ vertices such that their free space matrix with value $\delta$ is $\tloc$-local and $\tloc$ is known. Then, Algorithm~\ref{algo:frechet-tester-exact} is a Fréchet-tester that uses $\mathcal{O}(\frac{\tloc}{\eps}\log\frac{\tloc}{\eps})$ queries.
\end{theorem}
\begin{proof}
    If the Fréchet distance of $\firstcurve$ and $\secondcurve$ is at most $\delta$, there is a coupling path of cost $0$ through $\fsm$ and Algorithm~\ref{algo:frechet-tester-exact} always returns \enquote{yes}. 
    So we only need to determine the probability of a false positive.
    Suppose $\firstcurve$ and $\secondcurve$ are $(\eps,\delta)$-far and therefore, a minimum cost \restricteddiagonals\ path visits more than $\eps\n$ one-entries. If there are at least $\frac{\eps\n}{8\tloc}$ barrier-columns and barrier-rows, we return \enquote{no} with probability at least $\frac{9}{10}$ in Line~\ref{line:found-zero-col-row} by Lemma~\ref{lem:many-empty-rows-and-cols}. 
    If there are fewer than $\frac{\eps\n}{8\tloc}$ barrier-columns and barrier-rows and we get past Line~\ref{line:found-zero-col-row}, we return \enquote{no} with probability at least $\frac{9}{10}$ in Lines~\ref{line:permeability-query-col} and~\ref{line:permeability-query-row} by the Lemma~\ref{lem:main-lemma}. So in any case, the probability for a false positive is at most $\frac{1}{10}$.
    The number of queries follows from Lemma~\ref{lem:number-queries}.
\end{proof}

\section{Locality-oblivious testing of the discrete Fréchet distance}\label{sec:frechet-unknown-t}

Our main goal in this section is to obtain a Fr\'echet-tester that does not require an additional input (the value of $t$). 
Our approach for this is to first run a \enquote{tester} to find an estimate for $t$ 
and then run \refalg{frechet-tester-exact} with that parameter. 
Unfortunately, this combination is more complicated than just running two testers, for a number of different reasons:
First, \refalg{frechet-tester-exact} assumes that the matrix is $t$-local whereas the best we can do 
is to show that with some probability close to 1, it is \enquote{close} to being $t$-local.
This change breaks the proof of correctness of \refalg{frechet-tester-exact}.
The second complication is that we cannot plug in \textit{any} tester for locality, we need a specific one for our analysis that also guarantees some additional properties.
Guaranteeing the additional properties requires more queries than simply testing for locality.

\subsection{Testing the locality of the Free Space Matrix}\label{sec:locality}

To this end, we present a specific error model that is suitable for our purposes. 
In the following we define how $\tloc$-locality of a matrix can be violated. First, it could be that two columns or rows do not pass $\tloc$-locality according to Definition~\ref{def:pass-tloc}. In this case, we say that the respective columns or rows \emph{fail $\tloc$-locality}.
We introduce a second definition that considers a single column and looks for violations of $\tloc$-locality within that column and within the rows that have a zero-entry in this column. 

\begin{definition}\label{def:locfailure}
    A single column $i$ (resp. row $j$) \emph{fails second order $\tloc$-locality} if either it contains two zero-entries $(i,j),(i,j')$ (resp. $(i,j),(i',j)$) with $\abs{j'-j}>2\tloc$ (resp. $\abs{i'-i}>2\tloc$) or if there exists a zero-entry $(i,j)$ and two zero-entries $(i_1,j)$ and $(i_2,j)$ 
    (resp. $(i,j_1)$ and $(i,j_2)$) such that $|i_1-i_2| > 2t$ (resp. $|j_1-j_2|>2t$).
    Otherwise, it \emph{passes second order $\tloc$-locality}.
\end{definition}
Intuitively, a single column fails second order $\tloc$-locality if it has a zero-entry in a row that also has two zero-entries more than $2t$ distance apart
(similar intuition holds for the rows).  See \reffig{locfailure}.
Observe that if a matrix is not $t$-local then there either exists a pair of columns or rows that fail $\tloc$-locality or a column or row that fails second order $\tloc$-locality.

\begin{figure}[h]
    \centering
    \includegraphics[width=0.5\linewidth]{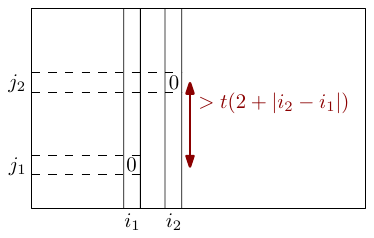}
    \caption{(left) Two columns fail $\tloc$-locality. (right) A column $i$ fails second order $\tloc$-locality.}
    \label{fig:locfailure}
\end{figure}

We now want to define what it means for a matrix to be $\zeta$-close to $\tloc$-local. This definition is specifically designed to fit the Fréchet-tester in Section~\ref{sec:frechet-unknown-t}. Other, maybe more natural definitions of $(\tloc,\zeta)$-close may need different property-testers. However, testing $\tloc$-locality is not our main focus in this paper. 

\begin{definition}[strongly $(t,\zeta)$-local]
    \label{def:stronglocal}
    We say that a free space matrix $\fsm$ is \emph{strongly $(\tloc,\zeta)$-local} for some $\tloc\geq1$ and $\zeta\in[0,1]$ if 
    we can partition the set of all columns and rows into two disjoint sets $W$ and $I$ with $\abs{I}=\zeta\n$ such that the following properties hold:
    i) columns $1,\n$ and rows $1,\n$ are in $W$, ii) any two columns and any two rows in $W$ pass $\tloc$-locality and iii) any row or column in $W$ passes second order $\tloc$-locality. We call the set $W$ the \emph{witness set} and the columns and rows in $I$ \emph{ignored}.  
\end{definition}

If there are $(1-\zeta)\n$ rows and $(1-\zeta)\n$ columns such that no two of these rows or columns fail $\tloc$-locality and no column or row fails second order $\tloc$-locality, then they serve as a witness set to see that the matrix is $(\tloc,2\zeta)$-local.
Because of this, we will check for failures of $\tloc$-localities and second order $\tloc$-localities among the columns and then among the rows in a similar fashion. 

\begin{observation}\label{obs:checks-for-locality}
    When we query two columns $i$ and $i'$, we can check if they fail $t$-locality by comparing the lowest zero-entry in column $i$ with the highest zero-entry in column $i'$ and vice versa.
    Also, after querying a column $i$, testing for second order $\tloc$-locality can be done with at most $2\tloc+1$ additional row queries at all zero-entries of the column $i$.
\end{observation}

Our testing algorithm is quite simple and, in fact, it is an extension of the classic \enquote{monotonicity} testing in an array~\cite{bhattacharyya2022property}. The algorithm is given in Algorithm~\ref{algo:t-locality-tester}.

\begin{algorithm}
    \caption{Locality-tester($\fsm$,$\sigma$,$\tloc$)}
    \label{algo:t-locality-tester}
    \begin{enumerate}
        \item $T\gets$ binary search tree of height $\lceil\log_2\n\rceil$ on $[\n]$ with root $r=\lceil\frac{\n}{2}\rceil$.
        \item \textbf{repeat} $\lceil\frac{3}{\sigma}\rceil+2$ times:
        \item ~~First two iterations: $i\gets1$ and $i\gets\n$. Later: Choose $i\in[2,n-1]$ unif.\ at random.\label{line:sample-index-locality-tester}\;
        \item ~~\textbf{If} column $i$ or row $i$ fails second order $t$-locality \textbf{then} return \enquote{no}.\label{line:break-locality-in-one-row-or-col}\;
        \item ~~\textbf{foreach} $j$ on the $r$-$i$-path in $T$ \textbf{do}:
        \item ~~~~ \textbf{If} columns $i$ and $j$ fail $\tloc$-locality \textbf{then} return \enquote{no}.\label{line:column-locality-query}\;
        \item ~~~~ \textbf{If} rows $i$ and $j$ fail $\tloc$-locality \textbf{then} return \enquote{no}.\label{line:row-locality-query}\;
        \item return \enquote{yes}.\;
    \end{enumerate}
\end{algorithm}

\begin{lemma}\label{lem:locality-tester-decision}
    Given $\sigma>0$ and $\tloc\geq1$, Algorithm~\ref{algo:t-locality-tester} returns \enquote{yes} if $\fsm$ is
    $\tloc$-local; if $\fsm$ is not strongly $(2\tloc,\sigma)$-local,
    it returns \enquote{no} with probability at least $\frac{9}{10}$. It uses $\mathcal{O}(\frac{t +\log n}{\sigma})$ queries.
\end{lemma}
\begin{proof}
    We observe that there are no false negatives: If $\fsm$ is $\tloc$-local, then all rows and columns pass $\tloc$-locality and second order $\tloc$-locality and therefore, the algorithm returns
    \enquote{yes}. 
    So we want to analyze the probability of a false positive. 
    For this, we use the notion of good and bad columns and rows. We call a column $i$ \emph{bad} if $i$ is not second order
    $\tloc$-local or if there exists an ancestor $j$ of $i$ in $T$ as defined in Algorithm~\ref{algo:t-locality-tester} such that $i$ and $j$ fail $\tloc$-locality. 
    Otherwise, we call the column \emph{good}. Good and bad rows are defined analogously. Note that as soon as we sample a bad row or bad column, the algorithm returns \enquote{no}.
    We now show that no two good columns fail $2\tloc$-locality and then in combination with the rows derive that the algorithm will return \enquote{no} with sufficiently high probability if $\fsm$ is not
    strongly $(2\tloc,\sigma)$-local. 

    \begin{claim}\label{claim:good-columns-dont-fail-locality}
        If $i_1$ and $i_2$ are good columns, then $i_1$ and $i_2$ pass $2\tloc$-locality.
    \end{claim}
    
    \begin{proof}[Proof of Claim~\ref{claim:good-columns-dont-fail-locality}]
    Let $i$ be the lowest common ancestor of $i_1$ and $i_2$ in the binary search tree $T$. If $i=i_1$, we know that $i_1$ and $i_2$ pass $\tloc$-locality (and therefore also $2\tloc$-locality) because $i_1$ is an ancestor of $i_2$. The same holds if $i=i_2$. W.l.o.g.\ assume that $i_1<i_2$. If $i_1<i<i_2$, let $j^{\ell}$ and $j^h$ be the lowest and highest zero-entry in column $i$. Analogously let $j_1^{\ell}$, $j_1^h$, $j_2^{\ell}$ and $j_2^h$ be the lowest and highest zero-entries in columns $i_1$ and $i_2$. Since $i_1$ and $i$ pass $\tloc$-locality, we have $j_1^{\ell}\geq j^h-(2\tloc+\tloc(i-i_1)$. Since $i$ and $i_2$ pass $\tloc$-locality, we have $j_2^h\leq j^{\ell}+(2\tloc+\tloc(i_2-i)$. Combining this yields 
    \[j_2^h-j_1^{\ell}\leq j^{\ell}-j^h+4\tloc+\tloc(i_2-i_1)\leq2\tloc(2+i_2-i_1).\]
    The last inequality holds because $j^{\ell}\leq j^h$.
    Similarly, we can show that $j_1^h-j_2^l\leq2\tloc(2+i_2-i_1)$ and therefore, $i_1$ and $i_2$ pass $2\tloc$-locality by Observation~\ref{obs:checks-for-locality}. 
    \end{proof}
    
    This claim holds analogously for rows. Suppose that $\fsm$ is not strongly $(2\tloc,\sigma)$-local. We want to bound the probability that we return \enquote{yes} in this case. If one of rows $1,\n$ or columns $1,\n$ is bad, we return \enquote{no} since we test $i=1,\n$ in the algorithm. So assume they are all good. Define $B_c$ to be the set of bad columns and $B_r$ the set of bad rows. 
    Define $I\coloneqq B_r\cup B_c$. 
    If $\abs{I}\leq\sigma\n$, the set of rows and columns without $I$ would be a witness set to see that $\fsm$ is strongly $(2\tloc,\sigma)$-local. Since by assumption this is not the case, we must have $\abs{I}>\sigma\n$.
    
    The probability that the algorithm samples a bad column in one iteration, where $i$ is chosen uniformly at random is at least
    $\frac{\abs{B_c}}{\n}$. Similarly, the probability that it samples a bad row in one iteration, where $i$ is chosen uniformly at random is at least $\frac{\abs{B_r}}{\n}$. 
    Thus, the probability that it never samples a bad row or column (the probability of a false positive) is at most
    \[\left(1-\frac{\abs{B_c}}{\n}\right)^{3/\sigma}\cdot\left(1-\frac{\abs{B_r}}{\n}\right)^{3/\sigma}\leq e^{-3(B_c+B_r)/(\sigma\n)}\leq e^{-3}\leq\frac{1}{10}.\qedhere\]
\end{proof}

The Locality-tester is given a value $\tloc$ as an argument. Next, we show that it is also possible to obtain a good estimate of the locality, as well.
The straightforward approach is to start with a small estimate for $\tloc$, e.g., $\tloc=1$, use
Algorithm~\ref{algo:t-locality-tester} and every time it fails, we can double the $\tloc$ and keep repeating. Our algorithm is very similar to this approach. The difference is merely that we use different $\zeta$ values throughout the iterations and we call the locality tester multiple times for the same $\tloc$. 

\begin{theorem}\label{thm:locality-optimization-second-variant}
    Suppose $\fsm$ is $\tloc^*$-local and we are given $\zeta>0$. We can design an algorithm that returns 
    a value $\tloc\leq2\tloc^*$ such that the matrix
    is strongly $(\tloc,\zeta/\tloc^2)$-local with probability at least $\frac{8}{9}$.
    The number of queries used is $\mathcal{O}(\frac{\log\log\tloc}{\zeta}(\tloc^3+\tloc^2\log\n))$. 
\end{theorem}
\begin{proof}
    The algorithm operates in rounds. 
    In the $i$-th round (starting from $i=1$) 
    we run Locality-tester($\fsm,\frac{\zeta}{2^{2(i+1)}},2^i)$ for $2(1+\log i)$ times. If the locality-tester returns \enquote{no} in any of
    these runs, we go to round $i+1$, otherwise, we return $\tloc=2^{i+1}$.

    First, we examine the number of queries. Let $i^*=\lceil\log\tloc^*\rceil$. Then, the algorithm terminates before or in round $i^*$ since no two columns or rows fail $2^i$-locality for any $2^i\geq\tloc^*$ and all rows and columns pass second order $\tloc^*$-locality. So we have at most $i^*$ rounds. 
    In round $i$, we perform $\mathcal{O}((2+2\log i)\cdot\frac{2^{2(i+1)}}{\zeta}\cdot(2^i+\log\n))$ queries by Lemma~\ref{lem:locality-tester-decision}. The total number of all queries from round $1$ to the last round (which is at most $i^*$) is $\mathcal{O}(\frac{\log\log\tloc}{\zeta}(\tloc^3+\tloc^2\log\n))$.
    
    By \reflem{locality-tester-decision}, 
    \enquote{no} answers are always correct.
    However, with probability at most $\frac{1}{10}$, the algorithm can return \enquote{yes} even though $\fsm$ is not
    $(2 \cdot 2^i,\frac{\zeta}{2^{2(i+1)}})$-local.  
    As all the executions of the algorithm are independent, 
    the probability that we get \enquote{yes} in all $2(1 + \log i)$ calls of the
    locality-tester in round $i$ is at most $10^{-(2(1+ \log i))} < \frac{1}{100 i^2}$. Using a union bound, the probability that we return a value $\tloc$ while $\fsm$ is not
    $(\tloc,\frac{\zeta}{\tloc^2})$-local is at most $\sum_{i=1}^{\infty} \frac{1}{100 i^2}<\frac{1}{9}$.
\end{proof}

\subsection{The combined algorithm}

We now present our algorithm for testing the Fréchet distance without knowing the locality parameter $\tloc$ in advance. Algorithm~\ref{algo:combined-frechet-tester}, uses the results from Section~\ref{sec:locality} to first estimate the locality of the matrix $\fsm$ and then uses the Fréchet-tester from Section~\ref{sec:frechet-with-known-t} to decide if the Fréchet distance of the two curves is smaller or equal to $\delta$.

\begin{algorithm}
    \caption{Fréchet-tester2($\fsm,\eps$)}
    \label{algo:combined-frechet-tester}
    \begin{enumerate}
        \item $\tloc\gets$ Output of the Algorithm in Theorem~\ref{thm:locality-optimization-second-variant} for $\fsm$ and $\zeta=\frac{\eps}{1600}$.\; 
        \item \textbf{return} Fréchet-tester1($\fsm,\tloc,\eps/3,\frac{4800\tloc^2}{\eps},2$).\;
    \end{enumerate}
\end{algorithm}

The following lemma is analogous to the main lemma (Lemma~\ref{lem:main-lemma}) from Section~\ref{sec:frechet-with-known-t}. For a proof of this lemma we refer to Section~\ref{sec:proofs-section-unknown-t}.

\begin{lemma}\label{lem:main-lemma-analogue}
    Let $\eps>0$ and let $\fsm$ be strongly $(\tloc,\frac{\eps}{1600\tloc^2})$-local and $\fsm[1,1]=0=\fsm[\n,\n]$. Suppose the total number of barrier-columns and barrier-rows is at most $\frac{\eps\n}{1600\tloc^2}$. Let $\optcoupling$ be a \restricteddiagonals\ path with lowest cost through $\fsm$ and suppose that $c(\optcoupling)>\eps\n$.
    Then, with probability at least $\frac{9}{10}$ at least one sampled interval in the set $\J$ in the call of Algorithm~\ref{algo:frechet-tester-exact} during Algorithm~\ref{algo:combined-frechet-tester} corresponds to a non-permeable block of columns or rows with probability at least $\frac{9}{10}$.
\end{lemma}

\begin{theorem}\label{thm:frechet-tester-unknown-t}
    Let $\delta>0$ and $0<\eps<2$ be given. Let $\firstcurve$ and $\secondcurve$ be curves with $\n$ vertices such that their free space matrix is \mbox{$\tloc^*$-local}. Then, Algorithm~\ref{algo:combined-frechet-tester} is a Fréchet-tester that uses 
    $\mathcal{O}(((\tloc^*)^3+(\tloc^*)^2\log\n)\frac{\log\log\tloc^*}{\eps})$ queries.
\end{theorem}
\begin{proof}
   The Locality-tester uses $\mathcal{O}(\frac{\log\log\tloc^*}{\eps}((\tloc^*)^3+(\tloc^*)^2\log\n))$ queries. Since $\tloc\leq2\tloc^*$, we need $\frac{(\tloc^*)^2}{\eps}$ queries to test for the barrier-columns and barrier-rows and then $\mathcal{O}\left(\frac{\tloc^*}{\eps}\log\left(\frac{\tloc^*}{\eps}\right)\right)$ queries for the permeability queries. Since $\frac{\tloc^*}{\eps}\in\mathcal{O}(\n^2)$, the running time is dominated by the Locality-tester.

    If the Fréchet distance of $\firstcurve$ and $\secondcurve$ is at most $\delta$, Algorithm~\ref{algo:combined-frechet-tester} will always return \enquote{yes}. So we have to bound the probability of false positives. Assume that $\firstcurve$ and $\secondcurve$ are $(\eps,\delta)$-far and therefore, a minimum cost \restricteddiagonals\ path $\optcoupling$ through $\fsm$ has cost higher than $\eps\n$.
    Suppose that $\fsm$ is strongly $(\tloc,\frac{\eps}{1600\tloc^2})$-local for the output $\tloc$ of the algorithm from Theorem~\ref{thm:locality-optimization-second-variant}. By the same theorem, this happens with probability at least $\frac{8}{9}$. 
    If the matrix contains more than $\frac{\eps\n}{1600\tloc^2}$ barrier-columns and barrier-rows, the subroutine that calls Algorithm~\ref{algo:frechet-tester-exact} returns \enquote{no} with probability at least $\frac{9}{10}$ by Lemma~\ref{lem:many-empty-rows-and-cols}. 
    If the matrix contains at most $\frac{\eps\n}{1600\tloc^2}$ barrier-columns and barrier-rows and the algorithm proceeds until the interval sampling, we are in the setting of Lemma~\ref{lem:main-lemma-analogue} and we return \enquote{no} with probability at least $\frac{9}{10}$.
    Therefore, the probability that the algorithm returns \enquote{no} is at least $\frac{8}{9}\cdot\frac{9}{10}=\frac{4}{5}$. So the probability of a false positive is at most $\frac{1}{5}$.
\end{proof}

\subsection{Proof of Lemma~\ref{lem:main-lemma-analogue}}\label{sec:proofs-section-unknown-t}

In our proof of Lemma~\ref{lem:main-lemma-analogue} we will group the path $\optcoupling$ into different groups just like we did in the proof of  Lemma~\ref{lem:main-lemma}. The next lemma makes sure that the column (and row) each group starts in is different to the column (and row) it ends in.
\begin{lemma}\label{lem:groups-end-in-different-rows-and-cols}
    Let $\fsm$ be strongly $(\tloc,\zeta)$-local with witness set $W$ for some $\zeta$. Let $(i,j)$ and $(i',j')$ be two zero-entries in $\fsm$ with $i,i'\in W$ (resp.\ $j,j'\in W$) and let $\pi$ be a coupling path from $(i,j)$ to $(i',j')$ that visits more than $2\tloc+1$ entries. Then we have $i\neq i'$ and $j\neq j'$.
\end{lemma}
\begin{proof}
    We prove the statement for $i,i'\in W$. The proof for $j,j'\in W$ is analogous.
    Since $i\in W$, column $i$ passes second order $\tloc$-locality. On the one hand, this means that any two zero-entries in column $i$ have distance at most $2\tloc$. Since $\pi$ starts and ends in a zero-entry and it visits more than $2\tloc+1$ entries in total, the first and last zero-entry cannot both be in column $i$ and we have $i\neq i'$.
    On the other hand, column $i$ passing second order $\tloc$-locality implies that all zero-entries in row $j$ have distance at most $2\tloc$. So we can apply the same logic to $j$ and $j'$.
\end{proof}

Our proof of Lemma~\ref{lem:main-lemma-analogue} uses the Impermeability Lemma (Lemma~\ref{lem:impermeability-lemma}) from Section~\ref{sec:frechet-with-known-t} and Lemma~\ref{lem:barrier-lemma-analogue}, which is analogous to the Barrier Lemma (Lemma~\ref{lem:barrier-lemma}) from Section~\ref{sec:frechet-with-known-t}. Recall that the Barrier Lemma gives a lower bound on the number of barrier-columns and barrier-rows in large subpaths of one-entries. The main work leading up to our proof of Lemma~\ref{lem:main-lemma-analogue} is to prove this Barrier Lemma analogue.
We start by introducing the definition of layers that helps us to group the zero-entries in the free space matrix. The layers correspond to the order in which the zero-entries can appear on a coupling path.
Let $(i,j)\neq(i',j')$ be zero-entries in the free space matrix. We say that $(i,j)$ \emph{dominates} $(i',j')$ if $i\geq j$ and $i'\geq j'$.
\begin{definition}[Layers]
    Let $R=[i,i']\times[j,j']$ be a rectangle in the free space matrix. We define the notion of layers in this rectangle. The first layer consists of all zero-entries in $R$ that do not dominate any other zero-entry in $R$. The layers are defined recursively such that all vertices on layer $i$ dominate at least one vertex on layer $i-1$ and dominate no vertex on layers $i$ and higher. This definition is visualized in Figure~\ref{fig:layers}.
\end{definition}

Note that every zero-entry in $R$ belongs to exactly one layer. So the layers group all zero-entries in $R$ into disjoint sets.

\begin{lemma}\label{lem:number-layers}
    Let $(i,j)$ and $(i',j')$ be zero-entries in $\fsm$ and let $R=[i,i']\times[j,j']$ be the rectangle spanned by these zero-entries. Let $L$ be the number of layers in $R$ and let $z$ be the number of zero-entries on a fixed minimum cost \restricteddiagonals\ path from $(i,j)$ to $(i',j')$. Then we have $z\leq L\leq 2z$.
\end{lemma}
\begin{proof}
    Let $\pi$ be a \restricteddiagonals\ path from $(i,j)$ to $(i',j')$ with the lowest cost
    and let $z$ be the number of zero-entries on it. 
    We observe that each zero-entry on $\pi$ dominates its predecessor. So we have $z\leq L$. 
    
    We construct a \restricteddiagonals\ path
    $\pi'$ from $(i,j)$ to $(i',j')$ that passes all layers and use this path to show that $L\leq2z$. 
    We observe that $(i,j)$ is the only zero-entry in layer $1$ because it is dominated by all other zero-entries in $R$ and similarly $(i',j')$ is the only zero-entry in layer $L$ because it dominates all other zero-entries in $R$.
    By the definition of the layers, $(i',j')$ dominates some zero-entry $(i_{L-1},j_{L-1})$ in layer $L-1$. This
    entry dominates some zero-entry on layer $L-2$ and so on. We have to end with a zero-entry $(i_2,j_2)$ in layer
    $2$ that dominates $(i,j)$. Then, we set $\pi'$ to be any \restricteddiagonals\ path from $(i,j)$ to $(i',j')$ that visits $(i_2,j_2),\ldots,(i_{L-1},j_{L-1})$ in that order. (See Figure~\ref{fig:layers}).
    Since any \restricteddiagonals\ path from $(i,j)$ to $(i',j')$ visits at most $i'-i+j'-j+1$ entries in total we have $c(\pi')\leq i'-i+j'-j+1-L$. Furthermore, we can bound $c(\pi)$ as follows:
    Let $(i,j)=(\ell_1,k_1),\ldots,(\ell_z,k_z)=(i',j')$ be the zero-entries on $\pi$. If $(\ell_{a+1},k_{a+1})\in\{(\ell_a+1,k_a),(\ell_a,k_a+1),(\ell_a+1,k_a+1)\}$, there is no one-entry in between, so the number of one-entries is at least $\ell_{a+1}-\ell_a+k_{a+1}-k_a-2$. Otherwise, the number of one-entries between $(\ell_a,k_a)$ and $(\ell_{a+1},k_{a+1})$ is exactly $\ell_{a+1}-\ell_a+k_{a+1}-k_a-1$ because $\pi$ is a \restricteddiagonals\ path. With this, we can derive that the total cost of $\pi$ is at least 
    \[\sum_{a=1}^{z-1}\ell_{a+1}-\ell_a+k_{a+1}-k_a-2=-2(z-1)+j'-j+i'-i.\]
    So in total we get
    \[-2(z-1)+j'-j+i'-i\leq c(\pi)\leq c(\pi')\leq i'-i+j'-j+1-L,\]
    which is equivalent to
    $L\leq2z-1<2z$.
\end{proof}

\begin{figure}
    \centering
    \includegraphics[width=0.25\linewidth]{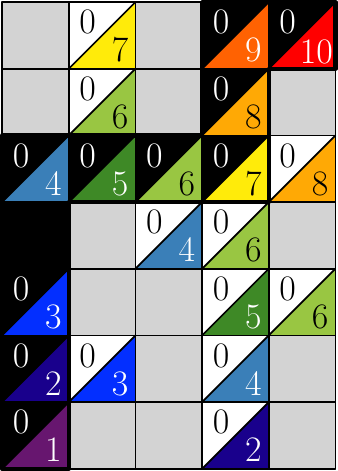}
    \hspace{1em}
    \raisebox{-4.25em}{\includegraphics[width=0.6\linewidth]{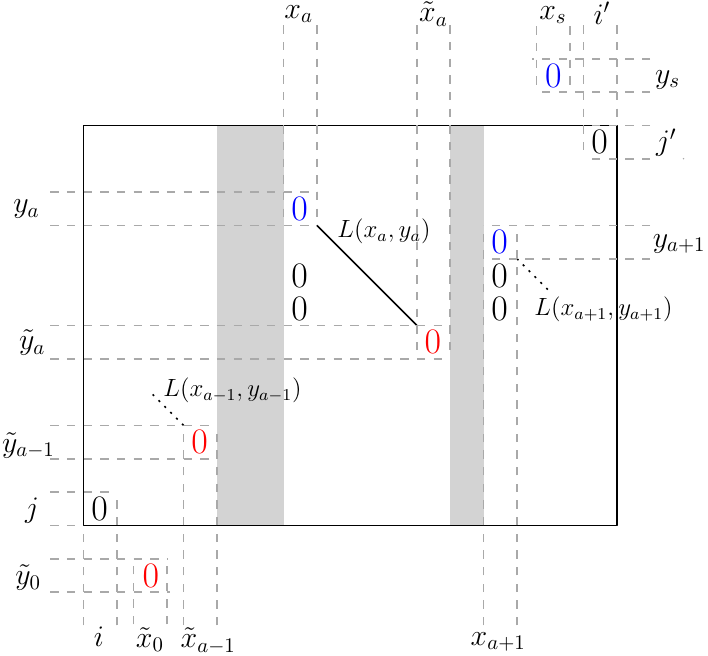}}
    \caption{Left: All zero-entries are shown in the upper corner of their cell. The respective layers are written in the lower corner of the cell and they are distinguished by color. The (black) shortest path contains exactly as many zero-entries as there are layers. Right: Visualization how $(x_a,y_a)$ and $(\Tilde{x}_a,\Tilde{y}_a)$ are chosen in Lemma~\ref{lem:bound-width-by-irrelevant-columns}. The depicted zero-entries that are in the sequence $(x_1,y_1),\ldots,(x_s,y_s)$ are blue and the those in the sequence $(\Tilde{x}_0,\Tilde{y}_0),\ldots,(\Tilde{x}_{s-1},\Tilde{y}_{s-1})$ are red. The gray columns are irrelevant and the layers are depicted as diagonal lines.}
    \label{fig:layers}
\end{figure}

We use the notion of layers in the proof of the next lemma. This lemma bounds the height (resp. width) of a subpath of the minimum cost \restricteddiagonals\ path that starts and ends in an unignored column (resp. row) and does not visit any zero-entries in unignored columns (resp. rows) in between. These subpaths will later play a role 
that is analogous to the subpaths of one-entries in the proof of the Lemma~\ref{lem:main-lemma}.

Before we get to the lemma, we introduce one more definition that is used throughout the whole analysis of the algorithm
\begin{definition}
    Let $\fsm$ be strongly $(\tloc,\zeta)$-local with witness set $W$ for some $\zeta$ and let $I$ be the set of ignored columns. We call a column or a row \emph{irrelevant} if it is either in $I$ or it is a barrier-column or barrier-row. Otherwise, we call it \emph{relevant}. 
\end{definition}

In contrast to the Barrier Lemma from Section~\ref{sec:frechet-with-known-t}, 
Lemma~\ref{lem:barrier-lemma-analogue}  will not give a lower bound just on the number of barrier-columns and barrier-rows, but on the number of irrelevant columns and rows. A key step toward our proof of the Barrier Lemma analogue is the following rather technical lemma that uses the idea of layers to bound the height (resp. width) of a rectangle spanned by two zero-entries with respect to the number of irrelevant columns (resp. rows) and the number of layers.

\begin{lemma}\label{lem:bound-width-by-irrelevant-columns}
    Suppose $\fsm$ is strongly $(\tloc,\zeta)$-local with witness set $W$ for some $\zeta$. Let $(i,j)$ and $(i',j')$ be zero-entries and denote by $m_c$ the number of irrelevant columns in $[i,i']$ and by $m_r$ the number of irrelevant rows in
    $[j,j']$. Denote by $L$ the number of layers in $R=[i,i']\times[j,j']$.  
    Then, we have the following statements:
    \begin{enumerate}
        \item If $i,i'\in W$, we have $\frac{j'-j}{\tloc}\leq3(L-1)+m_c.$\label{case:columns-relevant}
        \item If $j,j'\in W$, we have $\frac{i'-i}{\tloc}\leq3(L-1)+m_r.$\label{case:rows-relevant}
    \end{enumerate}
\end{lemma}
\begin{proof}
    We  prove the statement only for Case~\ref{case:columns-relevant}. Case~\ref{case:rows-relevant} can be proven analogously.
    For a zero-entry $(x,y)$, we denote by $L(x,y)$ the layer that the entry is on.

    Let $\Tilde{x}_0\leq i'$ be maximal such that column $\Tilde{x}_0$ is relevant and there exists a zero-entry $(\Tilde{x}_0,\Tilde{y}_0)$ with $\Tilde{y}_0\leq j$. Note that $(i,j)$ is a candidate for $(\Tilde{x}_0,\Tilde{y}_0)$. Define $(x_1,y_1),\ldots,(x_s,y_s)$ and $(\Tilde{x}_1,\Tilde{y}_1),\ldots,(\Tilde{x}_{s-1},\Tilde{y}_{s-1})$ recursively as follows: Given $\Tilde{x}_{a-1}$, we define $x_a$ to be the first relevant column after $\Tilde{x}_{a-1}$. Let $y_a$ be the row of the highest zero-entry in column $x_a$. If $y_a\geq j'$, then $(\Tilde{x}_{a-1},\Tilde{y}_{a-1})$ and $(x_a,y_a)$ are the last tuples in the respective sequences. If $y_a\leq j'$, we define $\Tilde{x}_a$ to be the largest relevant row that contains a zero-entry on layer $L(x_a,y_a)$, namely the zero-entry $(\Tilde{x}_a,\Tilde{y}_a)$. For a visualization of the sequences refer to Figure~\ref{fig:layers}.
    We can observe a set of useful properties of these sequences of indices:
    \begin{enumerate}
        \item It holds $\Tilde{y}_0\leq j$ and $y_s\geq j'$.\label{item:first-and-last-y}
        \item It holds $\Tilde{y}_a\leq y_a$ for all $a=1,\ldots,s-1$.\label{item:y-inequality}
        \item The columns $x_1,\ldots,x_s$ and $\Tilde{x}_0,\ldots,\Tilde{x}_{s-1}$ are all in $W$.\label{item:all-columns-unignored}
        \item There are $x_a-\Tilde{x}_{a-1}-1$ irrelevant columns between $\Tilde{x}_{a-1}$ and $x_a$.\label{item:irrelevant-columns-in-between}
        \item It holds $i\leq\Tilde{x}_0<x_1\leq\Tilde{x}_1<x_2\leq\ldots\leq\Tilde{x}_{s-1}<x_s\leq i'$.\label{item:x-sorted}
        \item The layers $L(x_1,y_1),\ldots,L(x_{s-1},y_{s-1})$ are all pairwise different.\label{item:pw-different-layers}
    \end{enumerate}
    
    The goal is to bound $j'-j$ from above. Because of Property~\ref{item:first-and-last-y}, we get 
    \[j'-j\leq y_s-\Tilde{y}_0=y_s+\sum_{a=1}^{s-1}(-\Tilde{y}_a+\Tilde{y}_a)-\Tilde{y}_0,\]
    where we added a zero sum in the last equality.
    Using Property~\ref{item:y-inequality}, we can bound this from above by
    \[y_s+\sum_{a=1}^{s-1}(-\Tilde{y}_a+y_a)-\Tilde{y}_0=\sum_{a=1}^{s}y_a-\Tilde{y}_{a-1}.\]
    Property~\ref{item:all-columns-unignored} yields that columns $\Tilde{x}_{a-1}$ and $x_a$ satisfy $\tloc$-locality and therefore $y_a-\Tilde{y}_{a-1}\leq\tloc(2+x_a-\Tilde{x}_{a-1})$ for all $a=1\ldots,s$. If we plug this into the inequality from above we get
    \[j'-j\leq\sum_{a=1}^{s}\tloc(2+x_a-\Tilde{x}_{a-1})=\tloc\left(3s+\sum_{a=1}^s(x_a-\Tilde{x}_{a-1}-1)\right).\]
    Because of Property~\ref{item:irrelevant-columns-in-between} and Property~\ref{item:x-sorted}, the sum on the right side counts pairwise different irrelevant columns. So, we have $\sum_{a=1}^s(x_a-\Tilde{x}_{a-1}-1)\leq m_c$.
    Last but not least, we bound $s$ with respect to $z$. For this, we argue that none of the layers $L(x_1,y_1),\ldots,L(x_{s-1},y_{s-1})$ is equal to the first or last layer. We already observed that the first layer only contains the entry $(i,j)$ in column $i<x_a$ for all $a=1,\ldots,s$. Similarly, we have $x_a<i'$ for all $a=1,\ldots,s-1$ and therefore none of the stated layers can contain $(i',j')$, the only entry of the last layer. This together with Property~\ref{item:pw-different-layers} yields that $s-1\leq L-2$. So, in total we get
    \[j'-j\leq\tloc(3(L-1)+m_c),\]
    which is equivalent to Case~\ref{case:columns-relevant} in the lemma.
\end{proof}

The following lemma underlines why the second order $\tloc$-locality is a useful property. From this we can derive that the zero-entries in an unignored column or row are not too far apart.
\begin{lemma}\label{lem:ignored-row-with-many-unblocked}
    Let $\fsm$ be strongly $(\tloc,\zeta)$-local for some $\zeta$ with witness set $W$. 
    \begin{enumerate}
        \item If column $i$ has two zero-entries $(i,j_1)$ and $(i,j_2)$ with $\abs{j_2-j_1}>2\tloc$, then $i$ is not in $W$ and all rows $j$ with a zero-entry $(i,j)$ are not in $W$ either. \label{case:column}
        \item If row $j$ has two zero-entries $(i_1,j)$ and $(i_2,j)$ with $\abs{i_2-i_1}>2\tloc$, then $j$ is not in $W$ and all columns $i$ with a zero-entry $(i,j)$ are not in $W$ either. \label{case:row}
    \end{enumerate}
\end{lemma}
\begin{proof}
We only prove Case~\ref{case:column}. The proof for Case~\ref{case:row} is analogous.
Since $\abs{j_2-j_1}>2\tloc$, column $i$ fails second order $\tloc$-locality and cannot be in $W$ by definition.
    Suppose $(i,j)$ is a zero-entry in column $i$. Then, column $i$ shows that row $j$ cannot pass second order $\tloc$-locality and therefore, $j$ cannot be in $W$.
\end{proof}

We are now ready to prove the Barrier Lemma analogue for the analysis of the second Fréchet-tester.  

\begin{lemma}[Barrier Lemma Analogue]\label{lem:barrier-lemma-analogue}
    Let $\fsm$ be strongly $(\tloc,\zeta)$-local with witness set $W$ for some $\zeta$. Let $\optcoupling$ be a \restricteddiagonals\ path from $(1,1)$ to $(\n,\n)$ with lowest cost. Let $(i,j),(i',j')\in\optcoupling$ be zero-entries such that $i$ and $i'$ (resp. $j$ and $j'$) are in $W$. Suppose $\optcoupling$ 
    visits no zero-entry in a column (resp. row) of $W$ between $(i,j)$ and $(i',j')$. Then we have a total of at least $\frac{j'-j}{50\tloc^2}$ (resp. $\frac{i'-i}{50\tloc^2}$) irrelevant columns in $[i+1,i'-1]$ and irrelevant rows in $[j+1,j'-1]$.
\end{lemma}
A visualization of this setting can be seen in the dashed box in Figure~\ref{fig:grouping-combined}.
\begin{proof}
    We prove the statement for $i,i'\in W$. The proof for $j,j'\in W$ is analogous. Since $\optcoupling$ is a \restricteddiagonals\ path from $(1,1)$ to $(\n,\n)$ with minimum cost, its subpath $\pi$ starting in $(i,j)$ and ending in $(i',j')$ is also a \restricteddiagonals\ path from $(i,j)$ to $(i',j')$ of minimum cost. Denote by $m_c$ the number of irrelevant columns between $i$ and $i'$ and by $m_r$ the number of irrelevant rows between $j$ and $j'$. Let $z$ be the number of zero-entries that $\pi$ visits. By Lemma~\ref{lem:number-layers} we know that there are at most $2z$ layers in the rectangle $R=[i,i']\times[j,j']$. We observe that we are now in the setting of Lemma~\ref{lem:bound-width-by-irrelevant-columns} Case~\ref{case:columns-relevant}. So we have 
    $\frac{j'-j}{\tloc}\leq3(2z-1)+m_c$. 

    If all columns or rows in $[i,i']$ and $[j,j']$ are in $W$, we can apply the Barrier Lemma (Lemma~\ref{lem:barrier-lemma}) to see that $\frac{j'-j}{50\tloc^2}\leq\frac{i'-i+j'-j}{2\tloc}\leq m_c+m_r$. 
    \begin{claim}\label{claim:bound-layers}
        Suppose there is at least one ignored column in $[i,i']$ or row in $[j,j']$. Then, it holds $z\leq 8\tloc(m_c+m_r)$.
    \end{claim}
    \begin{proof}[Proof of Claim~\ref{claim:bound-layers}]
        First of all, we see that each zero-entry has to be in an ignored column by assumption. So at most $m_c$ columns contain zero-entries on $\optcoupling$. Let $i_1,\ldots,i_s$ be ignored columns that have at least one zero-entry on $\optcoupling$. So $s\leq m_c$. For $i_k$, denote by $j^k_1,\ldots j^k_{n_k}$ all rows with a zero-entry on $\optcoupling$ in column $i_k$. So we have $z\leq 4\tloc+n_1+\ldots+n_s$ because each of columns $i$ and $i'$ can visit at most $2\tloc$ zero-entries on $\optcoupling$.
        If $n_k>2\tloc+1$, we know that all rows $j^k_1,\ldots,j^k_{n_k}$ are ignored by Lemma~\ref{lem:ignored-row-with-many-unblocked}. So all of these $n_k$ rows count into $m_r$. Additionally, we count at most $s-1$ rows twice, because this can  happen only for $j^k_{n_k}=j^{k+1}_1$. So we have
        \[z\leq 4\tloc+n_1+\ldots+n_s\leq4\tloc+(2\tloc+1)s+m_r+s\leq 4\tloc+(2\tloc+1)m_c+m_r+m_c.\]
        The claim follows since $m_c+m_r\geq1$, $m_c\leq m_c+m_r$ and $t\geq1$.
    \end{proof}
    We can plug this into the inequality from Lemma~\ref{lem:bound-width-by-irrelevant-columns} and get
    \[\frac{j'-j}{\tloc}\leq3(2z-1)+m_c\leq3( 16\tloc(m_c+m_r)-1)+m_c\leq50\tloc(m_c+m_r),\]
    which is equivalent to the statement in the lemma.
\end{proof}

\begin{figure}
    \centering
    \includegraphics[width=0.7\linewidth]{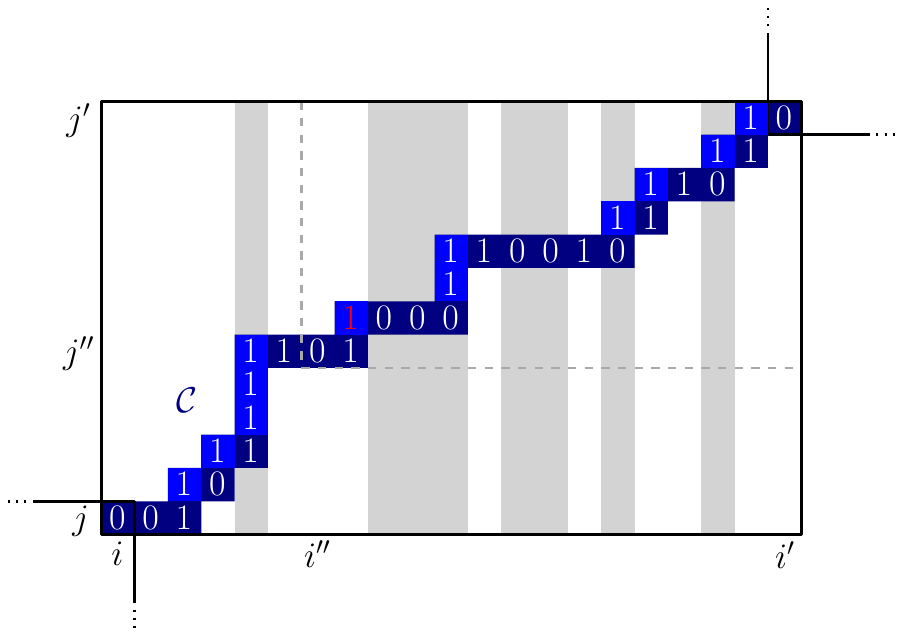}
    \caption{The row-entries are marked in lighter blue, the gray columns are ignored. Example for grouping in the proof of Lemma~\ref{lem:main-lemma-analogue} with $\tloc=2$: The group starts in $(i,j)$. The red row-entry is the $(2\tloc+2)$-th row-entry in the group. So we stop the group in $(i',j')$, the next zero-entry in an unignored column. Since it contains more than $4\tloc+2$ row-entries, this is a hollow group. The entries $(i'',j''),(i',j')$ satisfy the conditions in Lemma~\ref{lem:barrier-lemma-analogue}.}
    \label{fig:grouping-combined}
\end{figure}

With these lemmas we can prove an analogue to the main lemma (Lemma~\ref{lem:main-lemma}):
\begin{proof}[Proof of Lemma~\ref{lem:main-lemma-analogue}]
    Let $(i,j)$ be on the shortest \restricteddiagonals\ path $\optcoupling$ from $(1,1)$ to $(\n,\n)$. We say that $(i,j)$ is a \emph{row-entry} if it is a one-entry and its predecessor in $\optcoupling$ is $(i,j-1)$. Similarly, we say that $(i,j)$ is a \emph{column-entry} if it is a one-entry and its predecessor in $\optcoupling$ is $(i-1,j)$. 
    Note that for any one-entry $(i,j)$, its predecessor cannot be $(i-1,j-1)$ in any \restricteddiagonals\ path.
    Intuitively, the row-entries are the one-entries of $\optcoupling$ that are in a new row and the column-entries are the ones that are in a new column. A visualization of row-entries can be seen in Figure~\ref{fig:grouping-combined}.

    Let $b_r$ be the number of row-entries on $\optcoupling$ and let $b_c$ be the number of column-entries on $\optcoupling$. So we have $c(\optcoupling)=b_r+b_c$. Suppose that $b_r\geq b_c$ and therefore $b_r>\frac{\eps\n}{2}$. The other case can be done symmetrically by switching $\firstcurve$ and $\secondcurve$ and therefore switching columns and rows. Denote by $m_c$ the number of irrelevant columns and by $m_r$ the number of irrelevant rows in $\fsm$. So we have $m_c+m_r\leq\frac{\eps\n}{800\tloc^2}$.

    We now divide the path $\optcoupling$ into a minimum number of groups that each contain more than $2\tloc+1$ row-entries similarly to the grouping in the proof of the Lemma~\ref{lem:main-lemma}. A visualization can be seen in Figure~\ref{fig:grouping-combined}. The first group starts in $(1,1)$. We define the groups recursively. This time, we start the group in a zero-entry $(i,j)$ in an unignored column $i$. We add entries on $\optcoupling$ to the group until we have added $2\tloc+2$ row-entries or we arrived at $(\n,\n)$. Then, we finish the group in the next possible zero-entry $(i',j')$ that is in an unignored column $i'$. The next group starts again in $(i',j')$ if $(i',j')\neq(\n,\n)$. So all groups (except for possibly the last group) contain more than $2\tloc+1$ row-entries. This implies that for any group that goes from $(i,j)$ to $(i',j')$ (except for maybe the last group) we have that $i\neq i'$ and $j\neq j'$ by Lemma~\ref{lem:groups-end-in-different-rows-and-cols}. 

    We will first show that there are not too many groups that contain more than $4\tloc+2$ row-entries. Then, we will show that in the remaining groups there is at least a constant fraction of one-entries in small groups with no ignored rows and no ignored columns. On these groups we will use Lemma~\ref{lem:impermeability-lemma} to show that each of them induces at least one Permeablility query that fails. 

    Let $H_1,\ldots,H_k$ be the groups that contain more than $4\tloc+2$ row-entries. Let $b_1\ldots,b_k$ be the number of row-entries in these groups. We call these groups \emph{hollow} groups. All other groups are called \emph{dense}. An example of a hollow group can be seen in Figure~\ref{fig:grouping-combined}.
    \begin{claim}\label{claim:number-row-one-entries-in-hollow-groups}
        At most $\frac{\eps\n}{4}$ row-entries are contained in hollow groups.
    \end{claim}
    \begin{proof}[Proof of Claim~\ref{claim:number-row-one-entries-in-hollow-groups}]
        Let $H$ be one of the hollow groups with $b$ row-entries on the subpath of $\optcoupling$ in $H$. Let $(i',j')$ be the last entry of $H$ and therefore a zero-entry in an unignored column $i'$. Let $(i'',j'')$ be the last zero-entry in an unignored column $i''$ before $i'$. By the construction of $H$, we know that $\optcoupling$ visits at most $2\tloc+1$ row-entries before $(i'',j'')$. So it visits at least $b-2\tloc-1\geq\frac{b}{2}$ row-entries in between $(i'',j'')$ and $(i',j')$. We can apply Lemma~\ref{lem:barrier-lemma-analogue} to the segment of $\optcoupling$ from $(i'',j'')$ to $(i',j')$ and see that the total number of irrelevant columns in $[i'',i']$ and rows in $[j'',j']$ is at least $\frac{j'-j''}{50\tloc^2}\geq\frac{b-2\tloc-1}{50\tloc^2}\geq\frac{b}{100\tloc^2}$. If we sum over all of these, we count each irrelevant column at most once and each irrelevant row at most twice because of Lemma~\ref{lem:groups-end-in-different-rows-and-cols}. So in total we get
        \[m_c+m_r\geq\sum_{s=1}^k\frac{b_s}{200\tloc^2}.\]
        On the other hand, we have $m_c+m_r\leq\frac{\eps\n}{800\tloc^2}$. So in total, we have
        \[\sum_{s=1}^k\frac{b_s}{200\tloc^2}\leq\frac{\eps\n}{800\tloc^2},\]
        which is equivalent to
        \[\sum_{s=1}^kb_s\leq\frac{\eps\n}{4}.\qedhere\]
    \end{proof}
    Since there are at least $\frac{\eps\n}{2}$ row-entries and at most $\frac{\eps\n}{4}$ of them are in hollow groups, at least $\frac{\eps\n}{4}$ row-entries are contained in dense groups. Each of these groups contains at most $4\tloc+2\leq6\tloc$ row-entries. So there are at least $\lceil\frac{\eps\n}{24\tloc}\rceil$ dense groups. Let $G$ be a dense group that starts in $(i,j)$ and ends in $(i',j')$. We say that $i'-i+j'-j+2$ is the \emph{size} of $G$. This corresponds to the sum of the columns and rows that it visits. In total, every row and column is counted at most twice by the definition of the groups. So the sum of all group sizes is at most $4\n$. We observe that at least half of the groups (which are at least $\frac{\eps\n}{48\tloc}$) have size at most $2\cdot4\n/\lceil\frac{\eps\n}{24\tloc}\rceil\leq\frac{192\tloc}{\eps}$. We call these groups \emph{small}. This means that we have at least $\frac{\eps\n}{48\tloc}$ small dense groups and at most $\frac{\eps\n}{1600\tloc^2}\leq\frac{\eps\n}{96\tloc}$ ignored rows and columns. This means that at least half of the small dense groups do not contain any ignored column or row. So we have at least $\frac{\eps\n}{96\tloc}$ small dense groups that do not contain any ignored column or row. 

    We now want to create a set $X$ of intervals so that we can apply Lemma~\ref{lem:interval-sampling} to this set $X$. Each of the intervals $[i,i']$ in $X$ will have the property that the block of columns from $i$ to $i'$ is not permeable.
    
    Let $G$ be a small dense group that does not contain any ignored row or column and that starts in $(i,j)$ and ends in $(i',j')$. Since there are at least $2\tloc+1$ row-entries (unless $G$ is the very last group), we have $j'-j>2\tloc$. So we can add $[i,i']$ to $X$. By Lemma~\ref{lem:impermeability-lemma} we know that the block $[i,i']$ is not permeable. In the end, we have at least $\frac{\eps\n}{96\tloc}-1$ intervals in $X$, each of length at most $\frac{192\tloc}{\eps}\leq\frac{384\tloc}{\eps}$. By the definition of the groups, we have that each column is contained in one block only except for the first and last column, which are contained in at most two blocks. So, each column is contained in at most two intervals. 
    Thus, we can apply Lemma~\ref{lem:interval-sampling} to $X$ with  $\ell=\lceil\frac{384\tloc}{\eps}\rceil=\lceil\frac{128\tloc}{\eps/3}\rceil$, $K=\lceil\frac{\eps\n}{96\tloc}\rceil-1=\lceil\frac{\eps/3\n}{32\tloc}\rceil-1$, $c=2$ to see that the algorithm samples one of the intervals from $X$ with probability at least $\frac{9}{10}$.
    So if $b_r>b_c$, one of the permeability queries for columns fails with probability at least $\frac{9}{10}$. If $b_c>b_r$, we build the groups and the set $X$ with respect to the rows and then we can argue that one of the permeability queries for rows fails with probability at least $\frac{9}{10}$. So in total, we have that at least one permeability query fails with probability at least $\frac{9}{10}$.
\end{proof}

\section{Extensions for the discrete Fréchet distance}\label{sec:ext}

\subsection{Relaxing the input assumptions}\label{sec:largedelta}

Our results hold for $t$-local free space matrices. We argued in Section~\ref{sec:prelims} that this is a reasonable assumption for free space matrices of well-behaved curves. To be precise, our assumptions were two-fold: (1)  that the curves are $t$-straight (ii) that they have edge lengths within a constant factor of $\delta$~(see Lemma~\ref{lem:straight}).

In this section we relax the second assumption and show that our results still hold up to small approximation guarantees. In particular, we want to replace this assumption with a uniform sampling condition, namely that there exists a constant $\gamma \geq 1$, such that for any two edges $e$ and $e'$ of the input curve, it holds that $\frac{\ell(e)}{\gamma} \leq \ell(e') \leq \gamma \ell(e)$. However, the edge-lengths are  independent of $\delta$. For simplicity we assume the two input curves have the same complexity $n$, our proof can be easily extended to the unbalanced setting, where the curves have different complexities and different ranges of edge lengths.

\begin{definition}[$\mu$-approximate Fréchet-tester]
Assume we are given query-access to two curves $\firstcurve$ and $\secondcurve$, and we are given values $\delta > 0$ and $0<\eps<2$. If the two curves have discrete Fréchet distance at most $\delta$, we must return \enquote{yes},
if they are $(\eps,\mu\delta)$-far from each other w.r.t.\ the discrete Fréchet distance, the algorithm must return \enquote{no}, with probability at least $\succprob$. 
\end{definition}

Our main idea is to run the algorithms on a reduced set of entries of the free space matrix, which we define as follows.

\begin{definition}[$\beta$-reduced free space matrix]
For an $n \times n$ free space matrix $M$ and an integer parameter $\beta \geq 1$, we define the $\beta$-reduced free space matrix as the $ m\times m$  matrix with entries $ R[i,j] = M[i\beta, j\beta]$, where $m=\lfloor\frac{n}{\beta}\rfloor$.
\end{definition}

The next lemma quantifies the approximation guarantees of the reduced free space matrix with respect to the discrete Fréchet distance. 

\begin{lemma}
Let $\firstcurve$ and $\secondcurve$ be $t$-straight with edge lengths in $[\alpha/\gamma,\alpha]$ for some constant $\gamma \geq 1$. Let $\epsilon>0$ be a parameter and let $\beta=\max(1,\lfloor \frac{\eps\delta}{2\alpha}\rfloor)$. If there exists a coupling path $\mathcal{C}$ of cost zero in the $\beta$-reduced $\delta$-free space matrix then $\df{P}{Q} \leq (1+\eps)\delta$. If there exists no such path, then $\df{P}{Q} > (1-\eps)\delta$.
\end{lemma}
\begin{proof}
For $\beta=1$ the claim holds trivially true.
It remains to prove the claim in case $\lfloor \frac{\eps\delta}{2\alpha}\rfloor > 1$. In this case, we have $\alpha < \eps\delta/2$.

Assume there exists a coupling path $\mathcal{C}$ of cost zero as posited in the lemma. Consider a tuple $(i,j)\in\mathcal{C}$, and let $i'=i\beta$ and $j'=j\beta$ be the corresponding indices in the non-reduced free space matrix of $P$ and $Q$. For any $i_1 \in [i'-\beta + 1, i'+\beta -1]$ and 
$j_1 \in [j'-\beta + 1, j'+\beta -1 ]$ it holds by triangle inequality that 
\[ d(p_{i_1},q_{j_1}) \leq d(p_{i'},q_{j'}) + 2\alpha \beta = \delta + 2\alpha \lfloor \frac{\eps\delta}{2\alpha}\rfloor \leq (1+\eps)\delta \]
Therefore, we can find an extension of the coupling $\mathcal{C}$ in the non-reduced $(1+\eps)\delta$-free space matrix of $P$ and $Q$, and thus $\df{P}{Q}\leq (1+\eps)\delta$.
This shows the first part of the claim. 

As for the second part, assume $\df{P}{Q} \leq (1-\eps) \delta$. Then, there exists a coupling path $\mathcal{C'}$ of cost zero
in the non-reduced $(1-\eps)\delta$-free space matrix of $P$ and $Q$. 
Consider a tuple $(i,j)$ of the path. 
For any $i_1 \in [i-\beta + 1, i+\beta -1]$ and 
$j_1 \in [j-\beta + 1, j+\beta -1 ]$ it holds by triangle inequality that 
\[ d(p_{i_1},q_{j_1}) \leq d(p_{i},q_{j}) + 2\alpha \beta = (1-\eps) \delta + 2\alpha \lfloor \frac{\eps\delta}{2\alpha}\rfloor \leq \delta \]
Therefore, there exists a coupling path $\mathcal{C''}$ of cost zero in the $\beta$-reduced $\delta$-free space matrix of $P$ and $Q$.
By contraposition, this shows the second part of the claim. 
\end{proof}

Next, we prove an analogue of Lemma~\ref{lem:straight} for the reduced free space matrix.

\begin{lemma}\label{lem:straight2}
    Let $\firstcurve$ and $\secondcurve$ be $t$-straight with edge lengths in $[\alpha/\gamma,\alpha]$ for some constant $\gamma \geq 1$. Let $\epsilon>0$ be a parameter and let $\beta=\max(1,\lfloor \frac{\eps\delta}{2\alpha}\rfloor)$. The  $\beta$-reduced $\delta$-free space matrix is $\mathcal{O}(t/\eps)$-local.
\end{lemma}
\begin{proof}
    Let $R$ be the $\beta$-reduced $\delta$-free space matrix. Let $i_1,i_2,j_1,j_2$ be such that $R[i_1,j_1]=0=R[i_2,j_2]$ and $i_1\leq i_2$. Let $i'_1=\beta i_1$, $i'_2=\beta i_2$, $j'_1=\beta j_1$ , and $j'_2=\beta j_2$ be the corresponding indices in the non-reduced matrix.
    We have $\abs{i'_2-i'_1} =  \beta\abs{i_2-i_1}$ and 
    $\abs{j'_2-j'_1} = \beta\abs{j_2-j_1}$, since $\beta > 0$.     
    Moreover, 
    we have by the edge length constraint
    \begin{equation}
    d(q_{j'_1},q_{j'_2})\leq\ell(Q[j'_1,j'_2])\leq
    \alpha \abs{j'_1-j'_2}
    \end{equation}
    where $Q[j'_1,j'_2]$ is replaced by $Q[j'_2,j'_1]$ if $j'_2<j'_1$. 
    Secondly, we get from the edge length constraint and $t$-straightness that 
    \begin{equation}
        \frac{\alpha}{\gamma}\abs{i'_2-i'_1}\leq\ell(P[i'_1,i'_2])\leq t d(p_{i'_1},p_{i'_2})
    \end{equation}
    Together with the triangle inequality, this yields 
    \begin{eqnarray*}
     \abs{i_2-i_1} 
     &=&  \frac{1}{\beta} \abs{i'_2-i'_1} 
      \leq  \frac{\gamma t}{\beta \alpha}  d(p_{i'_1},p_{i'_2}) 
     \leq \frac{\gamma t}{\beta \alpha} (d(q_{j'_2},q_{j'_1})+ 2\delta) \\
      &\leq& \frac{\gamma t}{\beta \alpha} (\alpha \abs{j'_2-j'_1} + 2\delta) 
    = \frac{\gamma t}{\beta \alpha} ( \alpha \beta \abs{j_2-j_1} + 2\delta) \\
    &\leq& {\gamma t}(\abs{j_2-j_1}) +  2 \gamma t \frac{\delta}{\beta \alpha} 
    \end{eqnarray*}

    Since $\beta=\max(1,\lfloor \frac{\eps\delta}{2\alpha}\rfloor)$, it holds that  $\frac{\delta}{\beta \alpha} \in O(\frac{1}{\eps})$.
    Therefore, since $\gamma$ is a constant, we have derived
    \[  \abs{i_2-i_1} \leq t'(2+ \abs{j_1-j_2})\]
    for some $t'\in O(\frac{t}{\eps})$.  
    The other inequality is shown by reversing the roles of $P$ and $Q$.
\end{proof}

We can now modify the algorithms of Theorems~\ref{thm:frechet-tester-known-t} and Theorem~\ref{thm:frechet-tester-unknown-t} so that they only query a row or column if this row or column takes part in the $\beta$-reduced free space matrix for $\beta=\max(1,\lfloor \frac{\eps\delta}{2\alpha}\rfloor)$.

If we know that our input curves are $t$-straight, then we obtain from Theorem~\ref{thm:frechet-tester-known-t} an $(1+\eps')$-approximate Fréchet-Tester that performs at most $\mathcal{O}(\frac{\tloc}{\eps\cdot \eps'}\log\frac{\tloc}{\eps\cdot \eps'})$  queries.

Otherwise we use Theorem~\ref{thm:frechet-tester-unknown-t} and we obtain a $(1+\eps')$-approximate Fréchet-Tester that performs at most
$\mathcal{O}\left(\frac{\log\log\frac{\tloc}{\eps'}}{\eps}\left((\frac{\tloc}{\eps'})^3+(\frac{\tloc}{\eps'})^2\log\n\right)\right)$  queries.

\begin{theorem}\label{thm:frechet-tester-edge-lengts}
    Let $\delta>0$, $0<\eps<2$ and $\eps'>0$ be given. Let $P$ and $Q$ be $\tloc$-straight curves with edge lengths in $[\frac{\alpha}{\gamma},\alpha]$ for some constant $\gamma\geq1$. 
    If $\tloc$ is known, there is a $(1+\eps')$-approximate Fréchet-tester that performs at most $\mathcal{O}\left(\frac{\tloc}{\eps\cdot\eps'}\log\frac{\tloc}{\eps\cdot\eps'}\right)$ queries.
    If $\tloc$ is not known, there is a $(1+\eps')$-approximate Fréchet-tester that performs at most $\mathcal{O}\left(\left((\frac{\tloc}{\eps'})^3+(\frac{\tloc}{\eps'})^2\log\n 
    \right) \frac{\log\log\frac{\tloc}{\eps'}}{\eps}\right)$ queries.
\end{theorem}

\section{Additional results}\label{sec:additional}

\subsection{Testing  the continuous Fréchet distance}\label{sec:continuous}
This section describes how the Fréchet-testers for the discrete Fréchet distance can be used to serve as $(1+\eps')$-approximate Fréchet-testers for the continuous Fréchet distance. 
First, we define $(\eps,\delta)$-far for the continuous Fréchet-distance. Then, we explain how to use our Fréchet-testers on instances of the continuous Fréchet distance with this error model.

We first define the continuous Fréchet distance. We set $F_n$ to be the set of all continuous non-decreasing functions $f\colon[0,1]\rightarrow[1,n]$ with $f(0)=1$ and $f(1)=n$. Let $\firstcurve$ and $\secondcurve$ be polygonal curves with $\n$ and $\numvertsecond$ vertices. The \emph{continuous Fréchet distance} is defined to be
\[\delta_{\mathcal{F}}(P,Q)=\min_{f\in F_n,g\in F_m}\max_{\alpha\in[0,1]} d(P(f(\alpha)), Q(g(\alpha))).\]

Now we want to introduce a possible error model for the continuous Fréchet distance. Given a value $\delta>0$ and $f\in F_n$, $g\in F_m$, we define 
\[\mathcal{P}_{f,g}(P,Q)\coloneqq\int_{d(P(f(t)),Q(g(t)))>\delta}(||(P\circ f)'(t)||+||(Q\circ g)'(t)||)dt\]
as the partial difference, namely the total length of the portions of the two curves $P$ and $Q$ that are not matched with distance at most $\delta$ by $f$ and $g$. Building on this, we define the \emph{partial Fréchet difference} to be
\[\mathcal{P}_{\delta}(P,Q)\coloneqq\min_{f\in F_n,g\in F_m}\mathcal{P}_{f,g}(P,Q).\]
This is essentially the same as the partial Fréchet distance as defined in~\cite{buchin2009exact}.\footnote{Note that, in line with their definition, we assume the reparameterizations to be (piecewise) differentiable.}

\begin{definition}[$(\eps,\delta)$-far]
    Given two polygonal curves $P$ and $Q$, 
    we say that $P$ and $Q$ are $(\eps,\delta)$-far from one another w.r.t. the continuous Fréchet distance if $\mathcal{P}_{\delta}(P,Q)\geq\eps(\ell(P)+\ell(Q))$.
\end{definition}

\begin{definition}[$\mu$-approximate Fréchet-tester]
Assume we are given query-access to two curves $\firstcurve$ and $\secondcurve$, and we are given values $\delta > 0$ and $0<\eps<1$. If the two curves have continuous Fréchet distance at most $\delta$, we must return \enquote{yes},
if they are $(\eps,\mu\delta)$-far from each other w.r.t.\ the continuous Fréchet distance, the algorithm must return \enquote{no}, with probability at least $\succprob$. 
\end{definition}

In order to use the first Fréchet-tester introduced in this paper, we transform the polygonal curves into discrete curves by subsampling vertices along the edges of the curves. Denote by $P_a$ the discrete curve that arises from $P$ if we subsample vertices along the edges of $P$ such that $P$ visits a curve segment of length $a$ in between any two vertices for some $a\in(0,\ell(P)]$. 
The last edge might have length longer than $a$ but shorter than $2a$. 
Hence, $P_a$ has $\lfloor\frac{\ell(P)}{a}\rfloor$ edges. 
This yields that $\df{P_a}{Q_a}-2a<\delta_{\mathcal{F}}(P,Q)\leq\df{P_a}{Q_a}+2a$. 
If $P$ is $\kappa$-straight, so is $P_a$ for any $a\in(0,\ell(P)]$. 
Additionally, we can show that if $P$ and $Q$ are $(\eps,\delta)$-far from each other, then $P_{a}$ and $Q_a$ are $(\frac{\eps}{12},\delta-2a)$-far from each other.
\begin{lemma}\label{lem:far-continuous}
    Let $0<a\leq\min\{\ell(P),\ell(Q)\}$.
    If $P$ and $Q$ are $(12\eps,\delta+2a)$-far from each other, then $P_{a}$ and $Q_a$ are $(\eps,\delta)$-far from each other. 
\end{lemma}
\begin{proof}[Proof of Lemma~\ref{lem:far-continuous}]
    Denote by $n_a$ the number of vertices of $P_a$ and $m_a$ the number of vertices of $Q_a$.
    Suppose $P_a$ and $Q_a$ are not $(\eps,\delta)$-far from each other. This means that an optimum (diagonal) coupling $\optcoupling_a$ of them visits at most $\eps(n_a+m_a)$ tuples for which the respective points on $P_a$ and $Q_a$ have distance more than $\delta$. We can naturally transform this coupling into a traversal $T_a$ of $P$ and $Q$: Let $(x_i,y_i),(x_{i+1},y_{i+1})$ be consecutive in $\optcoupling_a$. If $x_i=x_{i+1}$ and $y_i\neq y_{i+1}$, then $T_a$ stays at $x_i$ on $P$ while it traverses the subcurve of $Q$ from $y_i$ to $y_{i+1}$ at unit speed. If $x_i\neq x_{i+1}$ and $y_i=y_{i+1}$, then $T_a$ stays at $y_i$ on $Q$ while it traverses the subcurve of $P$ from $x_i$ to $x_{i+1}$ at unit speed. If $x_i\neq x_{i+1}$ and $y_i\neq y_{i+1}$, then $T_a$ traverses both subcurves at constant speed and takes the same time for both traversals of their respective subcurves. (Note that a traversal can go at unit speed on both subcurves if they have the same length. The only time when this is not the case is if $x_{i+1}=P_a[n_a]$ or $y_{i+1}=Q_a[m_a]$.)

    Let $P_a[i,i+1]$ be an edge of $P_a$ such that the tuples $(p_i,y_k),(p_{i+1},y_{k+1})\in\optcoupling_a$ correspond to the traversal of this edge. We call $P_a[i,i+1]$ \emph{good} if $||p_i-y_k||\leq\delta$ and $||p_{i+1}-y_{k+1}||\leq\delta$, i.e. both adjacent tuples do not correspond to the at most $\eps(n_a+m_a)$ tuples on $\optcoupling_a$ that have distance larger than $\delta$. Otherwise, we call $P_a[i,i+1]$ \emph{bad}. The definitions of good and bad edges are analogous on $Q_a$.
    If $P_a[i,i+1]$ is good, the traversal of the edge in $T_a$ is within distance $\delta+2a$ of $Q$. The same holds for a good edge $Q_a[j,j+1]$. For most edges this is even within distance $\delta+a$. Only for the edges $P_a[n_a-1,n_a]$ and $Q_a[m_a-1,m_a]$ it can be larger.

    So the only edges that can get further away from the other curve than $\delta+2a$ are the bad edges. But there can be at most $4\eps(n_a+m_a)$ bad edges, each of length at most $2a$. This yields that the total length that $P$ and $Q$ spend at distance more than $\delta+2a$ is bounded by
    \[4\eps(n_a+m_a)2a\leq8\eps a(\frac{\ell(P)}{a}-1+\frac{\ell(Q)}{a}-1)\leq8\eps(\ell(P)+\ell(Q))+4\eps\cdot2a\leq12\eps(\ell(P)+\ell(Q)).\]
    The last inequality holds because $a\leq\ell(P)$ and $a\leq\ell(Q)$.
    This implies the statement of the lemma.
\end{proof}
We adjust our query oracle such that it can access the free space matrix of $P_a$ and $Q_a$. So we can apply the Fréchet-testers to $P_a$ and $Q_a$ for an adequate choice of $a$. We show that for given $\eps'>0$ and $0<\eps<2$, we can choose $a$ such that this yields a $(1+\eps')$-approximate Fréchet-tester for the continuous Fréchet distance of $P$ and $Q$.
Let $\eps''>0$.  For $a=\eps''\delta$, we test $P_{\eps''\delta}$ and $Q_{\eps''\delta}$ for the distance $\delta'\coloneqq(1+2\eps'')\delta$. We observe that $\fsm_{\delta'}$ is $\mathcal{O}(\frac{\kappa}{\eps''})$-local if $P$ and $Q$ are $\kappa$-straight. If $\delta_{\mathcal{F}}(P,Q)\leq\delta$, we have $\df{P_{\eps''\delta}}{Q_{\eps''\delta}}\leq\delta'$ and if $P$ and $Q$ are $(\eps,\mu\delta)$-far, we have that $P_{\eps''\delta}$ and $Q_{\eps''\delta}$ are $(\frac{\eps}{12},\delta')$-far for $\mu=1+4\eps''=1+\eps'$ for $\eps''=\eps'/4$. So, if we know $\kappa$, we can apply the first Fréchet-tester for $\frac{\eps}{12}$ and $\delta'$. If $\kappa$ is unknown, we apply the second Fréchet-tester for the same values. 

\begin{theorem}\label{thm:frechet-tester-continuous}
    Let $\delta, \eps, \eps' >0$ be given and assume $\eps<1$. Let $P$ and $Q$ be $\tloc$-straight curves with $\ell(P)=\ell(Q)$. If $\tloc$ is known, there exists a $(1+\eps')$-approximate Fréchet-tester w.r.t.\ the continuous Fréchet distance that performs at most $\mathcal{O}\left(\frac{\tloc}{\eps\cdot\eps'}\log\frac{\tloc}{\eps\cdot\eps'}\right)$ queries.
    If $\tloc$ is unknown, there exists a $(1+\eps')$-approximate Fréchet-tester w.r.t.\ the continuous Fréchet distance that performs at most 
    $\mathcal{O}\left(\left((\frac{\tloc}{\eps'})^3+(\frac{\tloc}{\eps'})^2\log\left\lceil\frac{\ell(P)}{\eps'\delta}\right\rceil\right)\frac{\log\log\frac{\tloc}{\eps'}}{\eps}\right)$ queries.
\end{theorem}

A natural consequence of the above theorem is that for $\tloc$-straight curves, the number of queries is bounded in terms of $n$ and the aspect ratio of the input, which we define as follows.

\begin{definition}[Aspect ratio]
For two polygonal curves $P$ and $Q$ and a distance threshold $\delta>0$, we define the \emph{aspect ratio} 
\[\phi\coloneqq\frac{\max_{u,v\in P\cup Q}d(u,v)}{\min\{\delta,\min_{u,v\in P\cup Q}d(u,v)\}}.\]
\end{definition}

We know that each edge of $P$ and $Q$ has length at most $\max_{u,v\in P\cup Q}d(u,v)$ and it trivially holds that $\delta\leq\min\{\delta,\min_{u,v\in P\cup Q}d(u,v)\}$. Therefore, if $P$ has $n$ edges, we have that $\frac{\ell(P_{\eps'\delta})}{\eps'\delta}\leq\frac{n\cdot \phi}{\eps'}$. This leads to the following corollary. Note that the only input constraint is that the curves are $\tloc$-straight and the aspect ratio is polynomial in $n$. There is no constraint on the edge lengths.

\begin{corollary}\label{cor:final}
    Let $\delta>0$, $0<\eps<1$ and $\eps'>0$ be given. Let $P$ and $Q$ be $\tloc$-straight curves with $n$ and $m\in\mathcal{O}(n)$ vertices, and aspect ratio $\phi \in O(\poly(n))$. There exists a $(1+\eps')$-approximate Fréchet-tester w.r.t.\ the continuous Fréchet distance that performs at most 
    $\mathcal{O}\left(\left((\frac{\tloc}{\eps'})^3+(\frac{\tloc}{\eps'})^2\log\left\lceil\frac{n}{\eps'}\right\rceil\right) \frac{\log\log\frac{\tloc}{\eps'}}{\eps}\right)$ queries.
\end{corollary}

\subsection{Testing the discrete Hausdorff distance}\label{sec:hausdorff}
For two curves $\firstcurve=\langle\vertfirstcurve_1,\ldots,\vertfirstcurve_{\numvertfirst}\rangle$ and $\secondcurve=\langle\vertsecondcurve_1,\ldots,\vertsecondcurve_{\numvertsecond}\rangle$, their \emph{discrete Hausdorff distance} is defined by 
\[\hd{\firstcurve}{\secondcurve}\coloneqq\max\{\overset{\rightarrow}{\mathcal{D}_{\mathcal{H}}}(\firstcurve,\secondcurve),\overset{\rightarrow}{\mathcal{D}_{\mathcal{H}}}(\secondcurve,\firstcurve)\},\]
where  $\overset{\rightarrow}{\mathcal{D}_{\mathcal{H}}}(\firstcurve,\secondcurve)\coloneqq\max_{i\in[n]}\min_{j\in[m]}\distance(\vertfirstcurve_i,\vertsecondcurve_j)$. 
For brevity, we simply call this the Hausdorff distance between $\firstcurve$ and $\secondcurve$.
Note that the Hausdorff distance between two curves is at most $\delta$ if and only if there are no columns or rows in their free space matrix $\fsm_{\delta}$ that only contain zeros.

\begin{definition}[$(\eps,\delta)$-far]
    Given two curves $\firstcurve$ and $\secondcurve$ consisting of $n$ and $m$ vertices each, we say that $\firstcurve$ and $\secondcurve$ are $(\eps,\delta)$-far from one another w.r.t.\ the Hausdorff distance if there exist a total of at most $\eps(n+m)$ barrier-columns and barrier-rows in their free space matrix with value $\delta$.
\end{definition}

\begin{definition}[Hausdorff-tester]
    Assume we are given two curves $\firstcurve$ and $\secondcurve$, and values $\delta>0$ and $\varepsilon\in(0,1)$. If the two curves have Hausdorff distance at most $\delta$, we must return \enquote{yes},
    if they are $(\varepsilon,\delta)$-far from each other w.r.t.\ the Hausdorff distance, the algorithm must return \enquote{no}, with probability at least $\succprob$.
\end{definition}

\begin{algorithm}
\caption{Hausdorff-tester($\fsm,\eps$)}
\label{algo:hausdorff-tester}
\begin{enumerate}
    \item Sample $\frac{2}{\eps}$ columns and $\frac{2}{\eps}$ rows uniformly at random  from $\fsm$.
    \item \textbf{If} one of them is a barrier-column or barrier-row \textbf{then} \textbf{return} \enquote{no}.
    \item \textbf{else} \textbf{return} \enquote{yes}.
\end{enumerate}
\end{algorithm}

\begin{theorem}\label{thm:hausdorff-tester}
    Algorithm~\ref{algo:hausdorff-tester} is a Hausdorff-tester that performs $\mathcal{O}(\frac{1}{\eps})$ queries.
\end{theorem}
\begin{proof}
    If the Hausdorff distance of $\firstcurve$ and $\secondcurve$ is at most $\delta$, the free space matrix $\fsm$ has no barrier-column or barrier-row. Hence, the algorithm always returns \enquote{yes} in this case.
    If $\firstcurve$ and $\secondcurve$ are $(\eps,\delta)$-far, there are at least $\eps(\numvertfirst+\numvertsecond)$ barrier-columns and barrier-rows. In that case, we either have at least $\eps\numvertfirst$ barrier-columns or at least $\eps\numvertsecond$ barrier-rows. W.l.o.g.\ we assume that there are at least $\eps\numvertfirst$ barrier-columns. When we sample one column, the probability that this is not a barrier-column is at most $1-\eps$. So the probability that the algorithm returns \enquote{yes} is at most $(1-\eps)^{2/\eps}\leq e^{-2}\leq\frac{1}{5}$. Hence, the probability that the algorithm returns \enquote{no} is at least $\frac{4}{5}$. 
\end{proof}

Note that we implicitly call a Hausdorff-tester in Algorithm~\ref{algo:frechet-tester-exact} when we test for barrier-columns and barrier-rows in the Fréchet-tester and hence the prior proof is very similar to the proof of Lemma~\ref{lem:many-empty-rows-and-cols}.

\subsection{A simple and fast \texorpdfstring{$t$}{t}-approximate Fréchet-tester}\label{subsec:apx-frechet-tester}

The goal of this section is to obtain a Fréchet-tester from the Hausdorff-tester from the previous section. 
While the algorithm is very simple (see Algorithm~\ref{algo:hausdorff-tester}), the analysis turns out to be more technically involved.
Since the proof needs the continuous Fréchet distance as well as the continuous Hausdorff distance, we introduce the latter here.
We consider $\firstcurve$ and $\secondcurve$ to be polygonal curves, i.e. the consecutive points $p_i,p_{i+1}$ are connected by the straight line segment from $p_i$ to $p_{i+1}$. Their \emph{continuous Hausdorff distance} is then defined by
\[\delta_{\mathcal{H}}(\firstcurve,\secondcurve)\coloneqq\max\{\overset{\rightarrow}{\delta_{\mathcal{H}}}(\firstcurve,\secondcurve),\overset{\rightarrow}{\delta_{\mathcal{H}}}(\secondcurve,\firstcurve)\},\]
where  $\overset{\rightarrow}{\delta_{\mathcal{H}}}(\firstcurve,\secondcurve)\coloneqq\max_{p\in\firstcurve}\min_{q\in\secondcurve}d(\vertfirstcurve,\vertsecondcurve)$. 
Note that it holds
$\df{P}{Q}\leq\delta_{\mathcal{F}}(P,Q)+E$ and $\delta_{\mathcal{H}}(P,Q)\leq\hd{P}{Q}+\frac{E}{2}$, where $E$ is the longest edge in $P$ and $Q$.

In~\cite{alt2004comparison} it is shown that $\delta_F\left(\firstcurve,\secondcurve\right)\leq(\kappa+1)\delta_H\left(\firstcurve,\secondcurve\right)$, if $\firstcurve$ and $\secondcurve$ are $\kappa$-bounded and $\max(d(p_1,q_1),d(p_n,q_m))\leq\delta_H\left(\firstcurve,\secondcurve\right)$.
We will show that a similar statement holds for the discrete Hausdorff distance and the discrete Fréchet distance.  

\begin{lemma}\label{lem:hausdorff-vs-frechet}
    For any pair of $\kappa$-straight polygonal curves $\firstcurve,\secondcurve$, we have
    \[\df{\firstcurve}{\secondcurve}\leq\left(\frac{3}{2}\kappa+\frac{5}{2}\right)\hd{\firstcurve}{\secondcurve},\]
    if $\max\{d(\vertfirstcurve_1,\vertsecondcurve_1),d(\vertfirstcurve_{\numvertfirst},\vertsecondcurve_{\numvertsecond})\}\leq\hd{\firstcurve}{\secondcurve}$ and all edges in $\firstcurve$ and $\secondcurve$ have length at most $\hd{\firstcurve}{\secondcurve}$.
\end{lemma}
\begin{proof}
Let $\mu$ be the length of the longest edge in $\firstcurve$ and $\secondcurve$. Then we have the following inequalities using the result from \cite{alt2004comparison}:
\[\df{\firstcurve}{\secondcurve}\leq\delta_{\mathcal{F}}(\firstcurve,\secondcurve)+\mu\leq(\kappa+1)\delta_{\mathcal{H}}(\firstcurve,\secondcurve)+\mu
\leq(\kappa+1)\hd{\firstcurve}{\secondcurve}+\frac{\kappa+3}{2}\cdot \mu.\]
By assumption, we have $\mu\leq\hd{\firstcurve}{\secondcurve}$ and thus the lemma is implied.
\end{proof}

We also need a statement that relates the property of being $(\eps,\delta)$-far under the discrete Hausdorff distance to 
the property of being $(\eps,\delta)$-far
under the discrete Fréchet distance. This will be done in Lemma~\ref{lem:far-frechet-hausdorff}. The lemma needs the following two helper lemmas that bound the distance of points on $P$ and $Q$ if there are zero-entries close to their indices in the free space matrix.

\begin{lemma}\label{lem:one-row}
    Let $P$ and $Q$ be $\kappa$-straight. If $\fsm_{\delta}[i,j]=0=\fsm_{\delta}[i,j']$, we have $d(p_i,q_s)\leq(\kappa+1)\delta$ for all $j\leq s\leq j'$. Similarly, if $\fsm_{\delta}[i,j]=0=\fsm_{\delta}[i',j]$, we have $d(p_k,q_j)\leq(\kappa+1)\delta$ for all $i\leq k\leq i'$.
\end{lemma}
\begin{proof}
    We only prove the first statement. The second statement can be proven analogously. For a visualization of the proof we refer to Figure~\ref{fig:apx-hausdorff}.
    It is given that $d(p_i,q_j)\leq\delta$ and $d(p_i,q_{j'})\leq\delta$. So by triangle inequality it holds $d(q_j,q_{j'})\leq2\delta$ and by $\kappa$-straightness, we have $\ell(Q[j,j'])\leq2\kappa\delta$. Let $j\leq s\leq j'$. Then, we have 
    \[\min\{d(q_j,q_s),d(q_{j'},q_s)\}\leq\min\{\ell(Q[j,s]),\ell(Q[s,j'])\}\leq\kappa\delta.\]
    Together with triangle inequality this implies    \[d(p_i,q_s)\leq\min\{d(p_i,q_j)+d(q_j,q_s),d(p_i,q_{j'})+d(q_{j'},q_s)\}\leq(\kappa+1)\delta.\qedhere\]
\end{proof}

\begin{figure}
    \centering
    \includegraphics[width=0.2\linewidth, page=1]{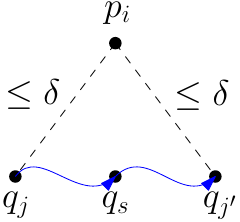}
    \hspace{5em}
    \includegraphics[width=0.25\linewidth, page=2]{dist-in-curves.pdf}
    \caption{Visualizations of the proofs of Lemma~\ref{lem:one-row} (left) and Lemma~\ref{lem:apx-hausdorff} (right).}
    \label{fig:apx-hausdorff}
\end{figure}

\begin{lemma}\label{lem:apx-hausdorff}
    Let $P$ and $Q$ be $\kappa$-straight and all edges in $P$ and $Q$ have length at most $\delta$. If $\fsm_{\delta}[i,j]=0=\fsm_{\delta}[i+1,j']$, we have $d(p_k,q_{\ell})\leq(\frac{3}{2}\kappa+2)\delta$ for all $k\in\{i,i+1\}$ and $j\leq\ell\leq j'$ or $j'\leq\ell\leq j$. Similarly, if $\fsm_{\delta}[i,j]=0=\fsm_{\delta}[i',j+1]$, we have $d(p_k,q_{\ell})\leq(\frac{3}{2}\kappa+2)\delta$ for all $\ell\in\{j,j+1\}$ and $i\leq k\leq i'$ or $i'\leq k\leq i$.
\end{lemma}
\begin{proof}
    We only prove the first statement. The second statement can be proven analogously. For a visualization of the proof we refer to Figure~\ref{fig:apx-hausdorff}.
    It is given that $d(p_i,q_j)\leq\delta$, $d(p_{i+1},q_{j'})\leq\delta$ and $d(p_i,p_{i+1})\leq\delta$. So by triangle inequality it holds $d(q_j,q_{j'})\leq3\delta$ and by $\kappa$-straightness, we have $\ell(Q[j,j'])\leq3\kappa\delta$. Let $j\leq s\leq j'$. Then, we have 
    \[\min\{d(q_j,q_s),d(q_{j'},q_s)\}\leq\min\{\ell(Q[j,s]),\ell(Q[s,j'])\}\leq\frac{3}{2}\kappa\delta.\]
    Together with triangle inequality this implies    \[d(p_i,q_s)\leq\min\{d(p_i,q_j)+d(q_j,q_s),d(p_i,p_{i+1})+d(p_{i+1},q_{j'})+d(q_{j'},q_s)\}\leq(\frac{3}{2}\kappa+2)\delta\]
    and
    \[d(p_{i+1},q_s)\leq\min\{d(p_{i+1},p_i)+d(p_i,q_j)+d(q_j,q_s),d(p_{i+1},q_{j'})+d(q_{j'},q_s)\}\leq(\frac{3}{2}\kappa+2)\delta.\qedhere\]
\end{proof}

The next lemma shows that if two $\kappa$-straight curves $P$ and $Q$ are $(\eps,\delta)$-far w.r.t.\ the discrete Fréchet distance, then they are also $(\mathcal{O}(\frac{\eps}{\kappa}),\mathcal{O}(\kappa\delta))$-far w.r.t.\ the discrete Hausdorff distance. 
\begin{lemma}\label{lem:far-frechet-hausdorff}
    Suppose $P$ and $Q$ are $\kappa$-straight with $\n$ and $\numvertsecond$ vertices, all edges in $P$ and $Q$ have length at most $\delta$ and the free space matrix $\fsm_{\delta}$ is $\tloc^*$-local.
    If there are at most $\eps'(\n+\numvertsecond)$ barrier-columns and barrier-rows in $\fsm_{\delta}$ and $\fsm_{\delta}[1,1]=0=\fsm_{\delta}[\n,\numvertsecond]$, the minimum cost coupling path from $(1,1)$ to $(\n,\numvertsecond)$ in $\fsm_{(\frac{3}{2}\kappa+2)\delta}$ has cost at most $4\tloc^*\eps'(\n+\numvertsecond)$.
\end{lemma}
\begin{proof}
    For a visualization of this proof, we refer to Figures~\ref{fig:coupling-fr-hd} and~\ref{fig:cases-hd-fr}.
    Define $\fsm\coloneqq\fsm_{\delta}$ and $\fsm'\coloneqq\fsm_{(\frac{3}{2}\kappa+2)\delta}$. Denote by $b_r$ the number of barrier-rows and by $b_c$ the number of barrier-columns in $\fsm$. Note that $b_r+b_c\leq\eps'(\n+\numvertsecond)$.
    We first define a coupling $\optcoupling$ and show that the corresponding coupling path visits at most $4\tloc^*(b_r+b_c)$ one-entries in $\fsm'$.

    Define $h(i)\coloneqq\max\{j\colon\exists i'\leq i\text{ with }\fsm[i',j]=0\}$. The value $h(i)$ represents the row with the \enquote{highest} zero-entry that exists to the left of $i$ in $\fsm$. We call $i$ a \emph{critical value} if $i=1$ or $h(i)\neq h(i-1)$ for $i>1$. We observe that $\fsm[i,h(i)]=0$ if $i$ is a critical value.
    Let $1=i_1<\ldots<i_k\leq\n$ be the critical values. Note that $h(i_k)=\numvertsecond$ since $\fsm[\n,\numvertsecond]=0$.
    We now construct the coupling $\optcoupling$ to visit all critical values $(i,h(i))$ and connect them via paths that go straight to the right from $(i_a,h(i_a))$ to $(i_{a+1},h(i_a))$ and then up to $(i_{a+1},h(i_{a+1}))$. More formally, we first take the sequence $(i_1,h(i_1)),\ldots,(i_k,h(i_k))$ to be $\optcoupling$. This is not necessarily a coupling yet. In the beginning we add the sequence $(1,1),\ldots,(1,h(i_1)-1)$, in the end we add the sequence $(i_k+1,h(i_k)),\ldots,(\n,h(i_k))$ and 
    in between $(i_a,h(i_a))$ and $(i_{a+1},h(i_{a+1}))$ we add the sequence
    \[(i_a+1,h(i_a)),\ldots,(i_{a+1},h(i_a)),\ldots,(i_{a+1},h(i_{a+1}-1).\]
    For a visualization, we refer to Figure~\ref{fig:coupling-fr-hd}.
    \begin{figure}
        \centering
        \includegraphics[width=0.35\linewidth]{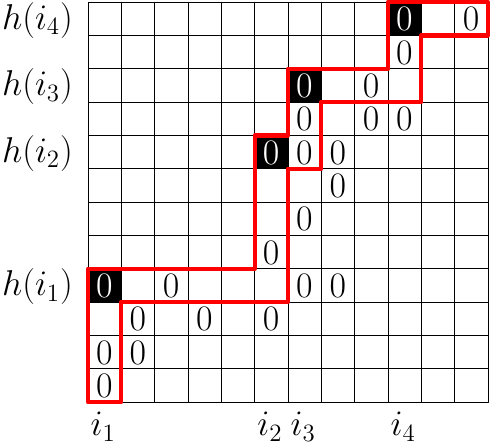}
        \caption{Visualization of the proof of Lemma~\ref{lem:far-frechet-hausdorff}: The zero-entries $(i_a,(h_{i_a}))$ for the critical values $i_a$ ($a\in\{1,2,3,4\}$) are depicted in black boxes. The coupling $\optcoupling$ is visualized in red.}
        \label{fig:coupling-fr-hd}
    \end{figure}
    We now show that this coupling satisfies the desired properties.
    Let $(i_a,h(i_a))$ be a critical tuple for some $1\leq a\leq k$.
    We first want to look at the entries of $\optcoupling$ below $(i_a,h(i_a))$. If $a=1$, we can apply Lemma~\ref{lem:one-row} to see that $d(p_1,q_j)\leq(\kappa+1)\delta$ for all $1\leq j\leq h(i_1)$ since $\fsm[1,1]=0=\fsm[1,h(i_1)]$ and hence, they are all zero-entries in $\fsm'$.
    So, consider $a>1$. For this, we distinguish two cases. Either column $i_a-1$ is a barrier-column in $\fsm$ or it is not. These cases are depicted in the upper row of Figure~\ref{fig:cases-hd-fr}.

    \begin{figure}[t]
    \centering
    \includegraphics[width=0.4\linewidth, page=3]{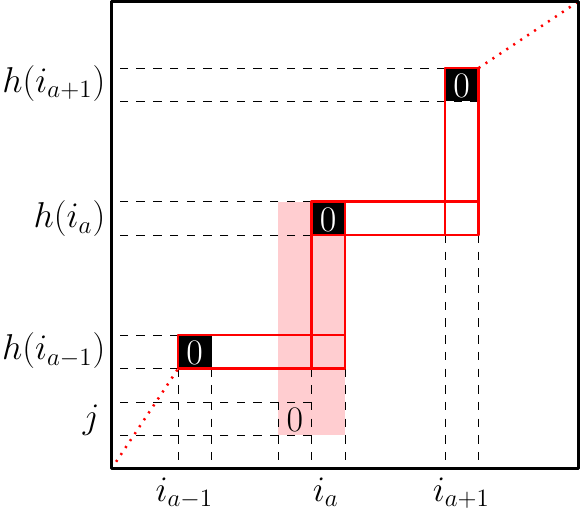}
    \hspace{1em}
    \includegraphics[width=0.4\linewidth, page=5]{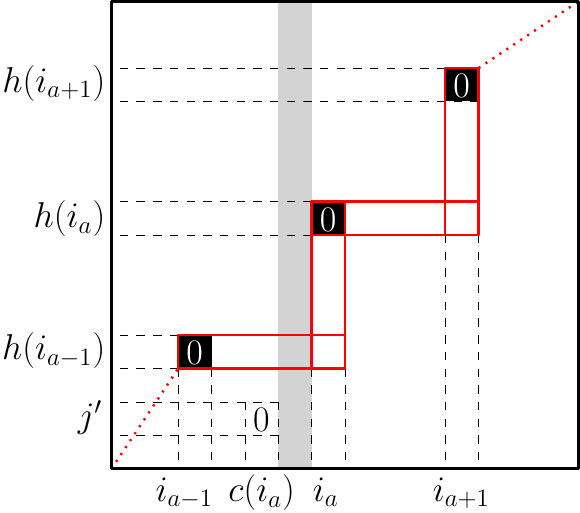}
    
    \vspace{1em}
    \includegraphics[width=0.4\linewidth, page=2]{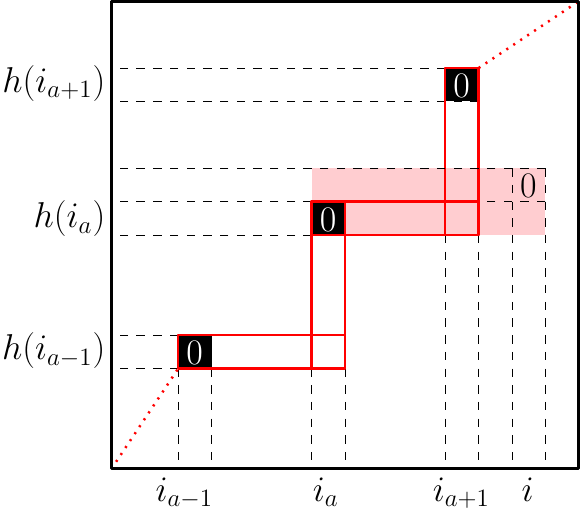}
    \hspace{1em}
    \includegraphics[width=0.4\linewidth, page=4]{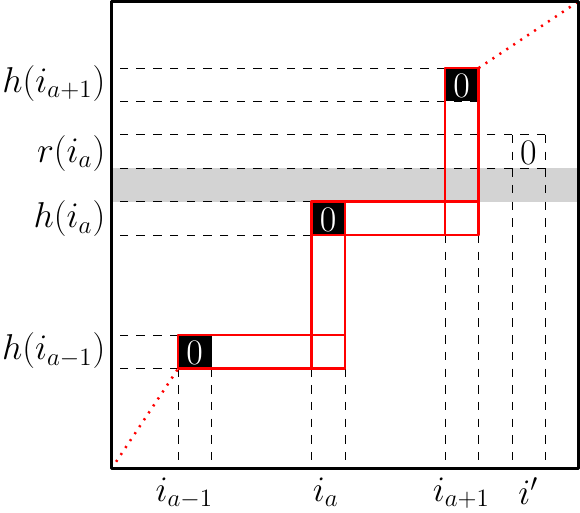}
    \caption{Visualization of the proof of Lemma~\ref{lem:far-frechet-hausdorff}: From upper left to lower right: The settings when $i_a-1$ is no barrier-column, $i_a-1$ is a barrier-column, $h(i_a)+1$ is no barrier-row and $h(i_a)+1$ is a barrier-row.
    Zero-entries of the critical values are depicted as black boxes. Cells shaded in pink are zero-entries in $\fsm'$. The gray columns and rows are barrier-columns and barrier-rows in $\fsm$.}
    \label{fig:cases-hd-fr}
\end{figure}

    If $i_a-1$ is not a barrier-column in $\fsm$, there is a zero-entry $(i_a-1,j)$ in $\fsm$ with $j\leq h(i_a-1)=h(i_{a-1})$. 
    Since any tuple $(i_a,\ell)\in\optcoupling$ satisfies $h(i_{a-1})\leq\ell\leq h(i_a)$,
    we can apply Lemma~\ref{lem:apx-hausdorff} to $(i_a-1,j)$ and $(i_a,h(i_a))$ to see that $d(p_{i_a},q_{\ell})\leq(\frac{3}{2}\kappa+2)\delta$ for all $(i_a,\ell)\in\optcoupling$. So these entries are zero-entries in $\fsm'$.

    If $i_a-1$ is a barrier-column in $\fsm$, let $c(i_a)<i_a-1$ be the largest index such that $c(i_a)$ is not a barrier-column in $\fsm$. Note that $i_{a-1}\leq c(i_a)$. So we know that there are at least $i_a-c(i_a)-1$ barrier-columns between $i_{a-1}$ and $i_a$ in $\fsm$. Let $(c(i_a),j')$ be a zero-entry in $\fsm$. By definition of $h(i_{a-1})$, we have $j'\leq h(i_{a-1})$. We can bound $h(i_a)-h(i_{a-1})\leq h(i_a)-j'\leq\tloc^*(2+i_a-c(i_a))$ with $\tloc^*$-locality. So we know that there are at most $\tloc^*(3+(i_a-c(i_a)-1))$ one-entries of $\fsm'$ below $(i_a,h(i_a))$ in $\optcoupling$. 
    Let $I_c$ be the set of critical values $i$ such that $i-1$ is a barrier-column in $\fsm$. Then, the number of one-entries of $\fsm'$ in $\optcoupling$ that are below an entry $(i,h(i))$ with $i\in I_c$ is at most
    \[\sum_{i\in I_c}\tloc^*(3+(i-c(i)-1))\leq 3\tloc^* b_c+\sum_{i\in I_c}\tloc^*(i-c(i)-1)\leq3\tloc^* b_c+\tloc^* b_c=4\tloc^* b_c.\]
    The first inequality holds because $\abs{I_c}\leq b_c$ and the second inequality holds because the sum counts pairwise different barrier-columns and there is at most $b_c$ of them.

    Now we look at the entries in $\optcoupling$ to the right of $(i_a,h(i_a))$. 
    If $a=k$, we can apply Lemma~\ref{lem:one-row} to see that $d(p_i,q_{\numvertsecond})\leq(\kappa+1)\delta<(\frac{3}{2}\kappa+2)\delta$ for all $i_k\leq i\leq\n$ since $\fsm[i_k,\numvertsecond]=0=\fsm[\n,\numvertsecond]$. So these entries are zero-entries in $\fsm'$. Consider $a<k$.
    As above, we distinguish the cases that $h(i_a)+1$ is a barrier-row in $\fsm$ or that it is not. These cases are depicted in the lower row of Figure~\ref{fig:cases-hd-fr}.

    If $h(i_a)+1$ is not a barrier-row in $\fsm$, there exists a zero-entry $(i,h(i_a)+1)$ in $\fsm$ with $i\geq i_{a+1}$ by the definition of $h(i_{a+1})$.
    Since any tuple $(k,h(i_a))\in\optcoupling$ satisfies $i_a\leq k\leq i_{a+1}$, we can apply Lemma~\ref{lem:apx-hausdorff} on $(i_a,h(i_a))$ and $(i,h(i_a)+1$ that $d(p_k,q_{h(i_a)})\leq(\frac{3}{2}\kappa+2)\delta$ for all $(k,h(i_a))\in\optcoupling$.
    So these entries are zero-entries in $\fsm'$.

    If $h(i_a)+1$ is a barrier-row in $\fsm$, let $r(i_a)> h(i_a)+1$ be the smallest index such that row $r(i_a)$ is not a barrier-row in $\fsm$. Let $(i',r(i_a))$ be a zero-entry in $\fsm$. By the definition of $h(i_{a+1})$, we have $i'\geq i_{a+1}$. So we know that there are at least $r(i_a)-h(i_a)-1$ barrier-rows in $\fsm$ between $h(i_a)$ and $h(i_{a+1})$. We can also bound $i_{a+1}-i_a\leq i'-i_a\leq\tloc^*(2+r(i_a)-i_a)$ using $\tloc^*$-locality. So we know that there are at most $\tloc^*(3+(r(i_a)-i_a-1))$ one-entries in $\fsm'$ to the right of $(i_a,h(i_a))$ in $\optcoupling$. 
    Let $I_r$ be the set of critical values $i$ such that $h(i)+1$ is a barrier-row in $\fsm$. Then, the number of entries one-entries to the right of any $(i,h(i))$ with $i\in I_r$ that $\optcoupling$ visits in $\fsm'$ is at most
    \[\sum_{i\in I_r}\tloc^*(3+(r(i)-i-1))\leq 3\tloc^* b_r+\sum_{i\in I_r}\tloc^*(r(i)-i-1)\leq3\tloc^* b_r+\tloc^* b_r=4\tloc^* b_r.\]
    The first inequality holds because $\abs{I_r}\leq b_r$ and the second inequality holds because the sum counts pairwise different barrier-rows and there is at most $b_r$ of them.
    So in total, the number of one-entries that $\optcoupling$ visits in $\fsm'$ is at most $4\tloc^*(b_r+b_c)\leq4\tloc^*\eps'(\n+\numvertsecond)$.
\end{proof}

\begin{theorem}\label{thm:apx-frechet-tester}
    Given $\delta>0$ and $\eps>0$, Algorithm~\ref{algo:hausdorff-tester} is an $\mathcal{O}(\kappa)$-approximate Fréchet-tester performing $\mathcal{O}(\frac{\kappa}{\eps})$ queries, if $\firstcurve$ and $\secondcurve$ are $\kappa$-straight, have edges with length at most $\delta$, have a $\mathcal{O}(\kappa)$-local free space matrix and satisfy $\max\{d(\vertfirstcurve_1,\vertsecondcurve_1),d(\vertfirstcurve_{\numvertfirst},\vertsecondcurve_{\numvertsecond})\}\leq\delta$.
\end{theorem}

Some of assumptions in the theorem above can be relaxed using the techniques from Lemma~\ref{lem:straight} and Section~\ref{sec:ext}. Since this is not our main result, we leave this as an exercise for the avid reader.

\section{Conclusions}
We proposed two property testing algorithms for the discrete Fréchet distance that work via query access to the  $\delta$-free space matrix of the input curves. For the analysis, we introduced the notion of $\tloc$-locality, a crucial property for our algorithms. The first algorithm assumes knowledge of the value of $\tloc$, and uses $\mathcal{O}(\frac{\tloc}{\eps}\log\frac{\tloc}{\eps})$ queries (see Theorem~\ref{thm:frechet-tester-known-t}). Surprisingly the number of queries is independent of $n$, the length of the input curves. The second algorithm does not require explicit knowledge of $\tloc$ and uses $\mathcal{O}(\frac{\log\log\tloc}{\eps}(\tloc^3+\tloc^2\log\n))$ queries (see Theorem~\ref{thm:frechet-tester-unknown-t}). We argued that the locality assumption is reasonable for approximate shortest paths. Moreover, we show that our algorithms can easily be extended to also test the continuous Fréchet distance (Theorem~\ref{thm:frechet-tester-continuous} and Corollary~\ref{cor:final}). 
We also introduced a simple Hausdorff-tester that performs $\mathcal{O}(\frac{1}{\eps})$ queries (Theorem~\ref{thm:hausdorff-tester}). Under some conditions, this Hausdorff-tester can also serve as a simple $\mathcal{O}(\tloc)$-approximate Fréchet-tester (Theorem~\ref{thm:apx-frechet-tester}), which can be proven using the fact that the Hausdorff distance approximates the Fréchet distance for approximate shortest paths.

For future work, it would be interesting to study other notions of realistic input curves, such as $c$-packed curves~\cite{cpacked}. Alas, the free space matrix of two $c$-packed curves is not necessarily $t$-local for any reasonable value of $t$, so a new approach may be needed. Another interesting question is if one can prove lower bounds for the problem. There are linear one-way communication complexity lower bounds for the DTW distance~\cite{dtw-lower} and there are  known reductions from variants of the set disjointness problem to the decision problem of the Fréchet distance~\cite{DriemelP21, MeintrupMR19}. It would be interesting to see if these results can be extended to show lower bounds in our setting.
Our algorithms may also lead to faster randomized algorithms for approximating the Fréchet distance for $t$-straight curves and in this setting can be compared to previous results~\cite{AronovHKWW06}. However, the error model, which allows an $\eps$-fraction of the input to be ignored or modified, may significantly distort the results and more ideas are needed to circumvent this. 

\bibliography{references}

\end{document}